\documentclass[twocolumn,pra,showpacs,amsmath,amssymb]{revtex4}
\usepackage{graphicx}
\usepackage{amssymb,amsthm,amsfonts,amstext}
\usepackage{url}
\usepackage[colorlinks, allcolors=blue]{hyperref}
\def\be{\begin{equation}}
\def\ee{\end{equation}}
\def\bea{\begin{eqnarray}}
\def\eea{\end{eqnarray}}
\def\bma{\begin{mathletters}}
\def\ema{\end{mathletters}}

\def\B{{\cal B}}

\def\M{{\cal M}}
\def\A{{\cal A}}
\def\0{\overline{0}}

\def\q0{\underline{0}}

\def\H{{\cal H}}

\def\S{S}
\def\C{{\mathbb C}}
\def\id{{\mathbb I}}

\def\H{{\cal H}}

\def\R{{\cal R}}

\def\A{{\cal A}}

\def\one{\leavevmode\hbox{\small1\normalsize\kern-.33em1}}

\def\bra#1{\langle#1|} \def\ket#1{|#1\rangle}
\def\braket#1#2{\langle#1|#2\rangle}

\def\proj#1{\ket{#1}\!\bra{#1}}
\def\ve#1{\langle#1\rangle}

\newtheorem{theo}{Theorem}

\newtheorem{lemma}[theo]{Lemma}

\newtheorem{remark}{Remark}
\def\id{{\mathbb I}}
\def\tr{\text{tr}}

\begin{document}

\title{Characterizing finite-dimensional quantum behavior}
\author{Miguel Navascu\'es$^{1}$, Adrien Feix$^{2,3}$, Mateus Ara\'ujo$^{2,3}$ and Tam\'as V\'ertesi$^{4}$}
\affiliation{$^1${\it Department of Physics, Bilkent University, Ankara 06800, Turkey}\\
$^2${\it Faculty of Physics, University of Vienna, Boltzmanngasse 5, 1090 Vienna, Austria}\\
$^3${\it Institute for Quantum Optics and Quantum Information (IQOQI), Boltzmangasse 3, 1090 Vienna, Austria}\\ $^4${\it Institute for Nuclear Research, Hungarian Academy of Sciences, H-4001 Debrecen, P.O. Box 51, Hungary}}

\begin{abstract}
We study and extend the semidefinite programming (SDP) hierarchies introduced in [Phys. Rev. Lett. 115, 020501] for the characterization of the statistical correlations arising from finite dimensional quantum systems. First, we introduce the dimension-constrained noncommutative polynomial optimization (NPO) paradigm, where a number of polynomial inequalities are defined and optimization is conducted over all feasible operator representations of bounded dimensionality. Important problems in device independent and semi-device independent quantum information science can be formulated (or almost formulated) in this framework. We present effective SDP hierarchies to attack the general dimension-constrained NPO problem (and related ones) and prove their asymptotic convergence. To illustrate the power of these relaxations, we use them to derive new dimension witnesses for temporal and Bell-type correlation scenarios, and also to bound the probability of success of quantum random access codes.
\end{abstract}

\maketitle

\section{Introduction}

Many problems in quantum information theory can be formulated as optimizations over operator algebras of a given dimensionality. Let us quickly review some of them.

In 1-way quantum communication complexity~\cite{review_comm,review_comm2}, two separate parties, call them Alice and Bob, are respectively handed the bit strings $x,y\in\{0,1\}^n$. Bob's task consists in guessing the value of the Boolean function $f(x,y)\in \{0,1\}$, and, to this aim, we allow Alice to send him a $D$-dimensional quantum system. Under these conditions, computing the maximum probability that Bob's guess is correct amounts to optimizing over all possible $D$-dimensional quantum states prepared by Alice and over all possible measurements conducted by Bob on such states.

In Bell scenarios~\cite{bell}, two or more distant parties conduct measurements over an unknown quantum state. It has been observed that, even if we do not assume any knowledge whatsoever about the mechanisms of the measurement devices, it is sometimes possible to lower bound the dimensionality $D$ of the quantum systems accessible to each party by virtue of the correlations between the measurement results alone~\cite{d_wit1,d_wit2,d_wit3,d_wit4}. In this regard, deriving \emph{dimension witnesses}, i.e., statistical inequalities satisfied by the correlations achievable through multipartite quantum systems of local dimension $D$, can be understood as an optimization over entangled states and measurement operators.

Entanglement distillation~\cite{distill}, or the capacity to prepare states close to a pure singlet given a number of mixed states through Local Operations and Classical Communication (LOCC), is one of the most conventional problems in quantum information science. More generally, determining whether the state transformation $\rho\to \sigma$ can be effected via LOCC can be interpreted as a feasibility problem, where the free variable is the corresponding LOCC map. If we restrict to local protocols or 1-way LOCC, the set of relevant maps admits a simple characterization in terms of tensor products of Krauss operators satisfying certain quadratic constraints. %However, even for such simplified scenarios, we lack general methods to limit the set of achievable transformations.

The above problems involve optimizations over a tuple of noncommuting variables $X_1,...,X_n$ satisfying a number of polynomial constraints, such as $X_i^2=X_i=X_i^\dagger$ (for projectors), or $X_iX_i^\dagger=X_i^\dagger X_i=\id$ (for unitaries). The number of total constraints is typically so low that, even fixing the dimensionality $D$ of the spaces where these operators act, we find a continuum of inequivalent representations.

Analogous problems emerge in the black-box approach to quantum information theory~\cite{bell,chsh,tsirel_pr,pr}, where the only constraints considered are essentially commutation relations between projection operators implemented by distant parties. The characterization of quantum nonlocality has boosted the field of noncommutative polynomial optimization (NPO) theory~\cite{NPAbound,npa,siam}, where the goal is, precisely, to conduct optimizations over all tuples of operators satisfying a number of polynomial inequalities. NPO theory achieves this via hierarchies of semi-definite programming (SDP) relaxations whose first levels approximate quite well the space of feasible solutions.

Unfortunately, NPO theory does not offer any means to bound or fix the dimension of the Hilbert spaces where such operators act. Since the afore-mentioned problems in quantum information theory become senseless or trivial in the high dimensionality limit, one would not expect NPO to be of any use.

This view changed with the publication of Ref.~\cite{finite_dim_short}, where a systematic way to devise hierarchies of SDP relaxations for a wide class of NPO problems under dimension constraints was introduced. Such relaxations, which seem to work quite well in practice, were used to derive a number of new results in quantum nonlocality and quantum communication complexity. Important theoretical aspects, such as the completeness of the hierarchies, or the explicit nature of the dimension constraints, were nonetheless left out. Actually, from a reading of Ref.~\cite{finite_dim_short}, it is not even clear which problems can be attacked with the new tools.

In this paper, we generalize the SDP schemes proposed in Ref.~\cite{finite_dim_short} to cover \emph{all} NPO problems where the dimensionality of the relevant Hilbert spaces is bounded. We will prove the convergence of the resulting SDP hierarchies and discuss their efficient implementation. Finally, we will use them to derive a number of new results in quantum information theory: from new bounds for quantum random access codes (QRAC) for both real and complex quantum systems to semi-device independent Positive Operator Valued Measure (POVM) detection~\cite{VB10,barra}; from the characterization of temporal correlations~\cite{temp_corr} under dimension constraints to the exploration of tripartite Bell scenarios where the dimensionality of just one of the parties is limited.

The structure of this paper is as follows: first, we will define the generic problem of NPO under polynomial constraints. Then, in Section~\ref{method}, we will present a hierarchy of SDP relaxations to tackle it. In Section~\ref{convergence} we will prove the convergence of this hierarchy. As it turns out, a straightforward implementation of the hierarchy would converge too slowly to be of much use, given reasonable computational resources. Hence in Section~\ref{exploit} we will give some hints to boost the speed of convergence of the method---to make it practical---, and exemplify its application by solving a specific problem on temporal correlations~\cite{temp_corr}. In Section~\ref{others} we will explore the performance of related SDP hierarchies to characterize quantum nonlocality under dimension constraints and quantum communication complexity. In Section~\ref{tips} we will offer some advice on how to code the corresponding programs. Finally, we will summarize our conclusions.

\section{Noncommutative Polynomial Optimization under dimension constraints}

Consider the set S of all $n$-tuples of self-adjoint operators $(X_1,..., X_n)$ satisfying the relations $\R=\{q_i(X)\geq 0: i = 1,...,m\}$. Here $q_i(X)$ denotes a Hermitian polynomial of the variables $X_1, ..., X_n$, while the notation $A\geq 0$ signifies that operator $A$ is positive semidefinite. We will call each feasible tuple $(X_1,..., X_n)$ a \emph{representation} of the polynomial relations $\R$. Given a Hermitian polynomial $p(X)$ and a natural number $D$, the problem we want to address is how to maximize the maximum eigenvalue of $p(X)$ over all representations of $\R$ of dimension $D$ or smaller. In other words, we want to solve the problem

\begin{align}
p^\star=&\max_{\H,X,\psi}\bra{\psi}p(X)\ket{\psi}\nonumber\\
&\mbox{s.t. }\mbox{dim}(\H)\leq D,q_i(X)\geq 0,\mbox{ for } i=1,...,m,
\label{npo_dim}
\end{align}

\noindent where the maximization is supposed to take place over all Hilbert spaces $\H$ with $dim(\H)\leq D$, all tuples of operators $(X_1,...,X_n)\subset B(\H)$ and all normalized states $\ket{\psi}\in \H$. Note that, if it were not for the dimension restriction $D$, the above would be a regular noncommutative polynomial optimization problem~\cite{siam}.

We say that the relations $\R$ satisfy the \emph{Archimedean condition} if there exist polynomials $f_j(X), g_{ij}(X)$ such that

\begin{multline}
C-\sum_{i}X_i^2= \sum_{j}f_j(X)^\dagger f_j(X)  \\ + \sum_{i,j}g_{ij}(X)^\dagger q_i(X)g_{ij}(X).
\label{archimedean}
\end{multline}

In the following, we provide a hierarchy of semidefinite programming relaxations for this problem. Such a hierarchy will provide a decreasing sequence of values
$p_1\geq p_2\geq...$ such that $p^k\geq p^\star,\forall k$. Moreover, if the Archimedean condition holds~\footnote{Actually, as we will see, a weaker condition suffices.}, then the hierarchy can be shown complete, i.e., $\lim_{k\to \infty} p^k = p^\star$.

\section{The method}
\label{method}

Let $y = (y_w)_{|w|\leq 2k}$ be a sequence of complex numbers labeled by monomials $w$ of the variables $X_1,...,X_n$ of degree $|w|$ smaller than or equal to $2k$. Such a sequence will be called a \emph{$2k^{th}$-order moment vector}. Given $y$, the $k^{th}$-order \emph{moment matrix} $M_k(y)$ is an array whose rows and columns are labeled by monomials of $X_1,..., X_n$ of degree at most $k$, and such that

\be
M_k(y)_{u,v} = y_{u^\dagger v}.
\ee

Given $y = (y_w)_{|w|\leq 2k}$ and a Hermitian polynomial $q(X) =\sum_w q_w w(X)$, where the $w$ in the summation ranges over all monomials of $X_1,...,X_n$ of degree at most $\text{deg}(q)$, the corresponding $k^{th}$-order \emph{localizing matrix} is defined as
\be
M_k(qy)_{u,v} =\sum_w q_w y_{u\dagger wv},
\ee

\noindent with $|u|,|v|\leq k-\left\lfloor \frac{\text{deg}(q)}{2}\right\rfloor$.

A sequence $y = (y_w)_{|w|\leq 2k}$ admits a \emph{quantum representation} if there exists a representation $(X_1, ..., X_n)\subset B(\H)$ of relations $\{q_i(X)\geq 0\}_i$, with $\mbox{dim}(\H)\leq D$, and a normalized vector $\ket{\psi}\in \H$ such that $y_w = \bra{\psi}w(X)\ket{\psi}$.
It is a standard result in non-commutative polynomial optimization theory that, if $(yw)$ admits a moment representation (of whatever dimensionality), then $M_k(y)$ and $M_k(q_iy)$ must be positive semidefinite matrices for all orders $k$~\cite{siam}.

The above positive semidefinite constraints are not dimension-dependent, and are actually obeyed by momenta emerging from representations of $\{q_i(X)\geq 0\}_i$ of arbitrary (even infinite) dimensionality. The key to introduce dimension constraints is to acknowledge that moment vectors $(y_w)$ admitting a quantum representation satisfy a number of extra linear restrictions depending on the value of $D$.

Some of such restrictions arise due to matrix polynomial identities (MPI)~\cite{MPI}: these are polynomials $s(X)$ of the variables $X_1,...,X_n$ which are identically zero when evaluated on matrices of dimensionality $D$ or smaller. For $D = 1$, all MPIs reduce to commutators, i.e., $[A_i,A_j] = 0$, if $A_i, A_j \in B(\C)$. Identifying $A_i = X_i$, this implies that sequences $y = (y_w)_{|w|\leq 2k}$ admitting a 1-dimensional moment representation must satisfy $y_{X_1X_2} - y_{X_2X_1}= 0$. Actually, for any value of $D$ there exist MPIs from which non-trivial linear constraints on $y$ can be derived. For $D = 2$, all MPIs are generated by composition of the identities $[[A_1,A_2]^2,A_3]=0$ and $\sum_{\pi\in S_4}\mbox{sgn}(\pi)A_{\pi(1)}A_{\pi(2)}A_{\pi(3)}A_{\pi(4)}=0$, where $S_4$ denotes the set of all permutations of four elements. The latter identity is a particular case of the family of polynomial identities $I_d$, with

\be
\sum_{\pi\in S_d}\mbox{sgn}(\pi)A_{\pi(1)}...A_{\pi(d)}=0.
\label{funda}
\ee

\noindent It can be proven that all $D\times D$ matrices satisfy $I_{2D}$~\cite{MPI}, also called the \emph{standard identity}. The problem of determining the generators of \emph{all} MPIs for dimensions $D$ greater than two is, however, open.

A non-trivial relaxation of problem~(\ref{npo_dim}) is thus:

\begin{align}
p^k=&\max_{y}\sum_wp_w y_w,\nonumber\\
\mbox{s.t. } &y\in \S^k_D, y_{\id}=1,M_k(y)\geq 0,\nonumber\\
&M(q_iy)\geq 0,\mbox{ for } i=1,...,m.
\label{npo_rel}
\end{align}

\noindent where $\S^k_D$ denotes the span of the set of feasible sequences $y = (y_w)_{|w|\leq 2k}$. This is a semidefinite program, and, as such, can be solved efficiently for moment matrices of moderate size (around $200\times 200$) in a normal desktop~\cite{sdp}.

Equivalently, we can re-express the positivity conditions as $\hat{M}_k\equiv M_k(y)\oplus\bigoplus_{i=1}^m M_k(q_iy)\geq 0$ and rewrite the objective function as a linear combination of the entries of the first diagonal block of $\hat{M}$. That way, we can regard the block-diagonal matrix $\hat{M}$ (and not $y$) as our free variable, hence arriving at the program

\begin{align}
p^k=&\max_{\hat{M}}\sum_wp_w \hat{M}_{w,\id},\nonumber\\
\mbox{s.t. } &\hat{M}\in \M^k_D, \hat{M}_{\id,\id}=1,\hat{M}\geq 0,
\label{npo_rel2}
\end{align}

\noindent where $\M^k_d$ denotes the span of the set of feasible extended moment matrices. This reformulation of problem~(\ref{npo_rel}), although conceptually more cumbersome, leads to simpler computer codes.

The key to implement either program is, of course, to identify the subspaces $\S^k_D,\M^k_D$. We will now provide two methods to do so. Both have advantages and disadvantages. In~\cite{mps_amazing} we provide yet a third method, which, although more complicated than the other two, requires considerably less memory and time resources, making it suitable for high order relaxations.

\vspace{10pt}
\noindent\textbf{The randomized method}

We sequentially generate $n$-tuples of random Hermitian $D\times D$ complex matrices $X^j\equiv (X^j_1,...,X^j_n)$ and normalized random vectors $\ket{\psi^j}\in \C^D$, which we use to build moment and localizing matrices $M^j_{u,v}=\bra{\psi^j}u(X^j)^\dagger v(X^j)\ket{\psi^j}$, $M(q_i)^j_{u,v}=\bra{\psi^j}u(X^j)^\dagger q_i(X) v(X^j)\ket{\psi^j}$, respectively. Their direct sum will constitute an extended moment matrix $\hat{M}^j$. Adopting the Hilbert-Schmidt scalar product $\langle A,B\rangle=\tr(A^\dagger B)$, one can apply the Gram-Schmidt process~\footnote{Or, better, a numerically stable variant, such as the modified Gram-Schmidt method~\cite{mod_GS}.} to the resulting sequence of feasible extended moment matrices in order to obtain an orthogonal basis $\tilde{M}^1,\tilde{M}^2,...$ for the space spanned by such matrices. We will notice that, for some number $N$, $\tilde{M}^{N+1}=0$, up to numerical precision. This is the point to terminate the Gram-Schmidt process and define the normalized matrices $\{\Gamma^i\equiv \frac{\tilde{M}^j}{\sqrt{\tr{(\tilde{M}^j)^2}}}:j=1,...,N\}$. It is easy to see that, even though the matrix basis $\{\Gamma^j\}_{j=1}^N$ was obtained randomly, the space it represents is always the same, namely, $\M_{D}^k$.

Indeed, let $N=\mbox{dim}(\M_{D}^k)$, and suppose that $\tilde{M}^1,...,\tilde{M}^{j-1}$ are non-zero, with $j\leq N$. Then the entries of the matrix $\tilde{M}^{j}$ will be polynomials of the components of $X^j,\psi^j$. Since $j\leq N$, there exists a choice of $\vec{z}^j$ such that $\tilde{M}^j(X^j,\psi^j)\not=0$. The probability that a non-zero polynomial vanishes when evaluated randomly is zero, and so we conclude that $\tilde{M}^j$ will be non-zero with probability 1. On the other hand, $\tilde{M}^1=\hat{M}\not=0$, so by induction we have that $N$ randomly chosen moment matrices will span $\M_{D}^k$ with certainty. Consequently, $\tilde{M}^{N+1}=0$ indicates when to stop the procedure.

\begin{remark}
\label{support}
A cautionary note is in order: for high orders $k$, it is expected that program~(\ref{npo_rel2}), as it is written, will not admit strongly feasible points. That is, the subspace $\M^k_D$ will not contain any positive definite matrix. This can be problematic, as many SDP solvers need strong feasibility to operate. The solution is to add up the random extended moment matrices $\hat{M}^j$ as we produce them, i.e., to compute the operator $T=\frac{1}{N}\sum_{j=1}^N\hat{M}^j$. Since $\{\hat{M}^j\}$ were randomly generated, it can be argued that the support of any matrix $\hat{M}\in\M^k_D$ is contained in the support of $T$. Let $V$ be any matrix mapping $\mathrm{supp}(T)$ to $\C^{\mathrm{dim}(\mathrm{supp}(T))}$ isometrically. We just need to replace the positivity condition $\hat{M}\geq 0$ in~(\ref{npo_rel2}) by $V \hat{M} V^\dagger\geq 0$, which, by definition, admits a stricly feasible point.
\end{remark}

This method has the advantage that it is extremely easy to program; more so when the constraints $\{q_i(X)\geq 0\}$ reduce to polynomial identities, as we will see. One disadvantage is that, in practice, the decision to stop the protocol amounts to verifying that the entries $\tilde{M}^{N+1}$ are zero up to $\epsilon$ precision. Choosing the right value for the threshold $\epsilon$ is a delicate matter: if too small, the algorithm will not halt; if too large, the algorithm will stop before it finds a complete basis for $\M_{D}^k$. The second problem is that, due to rounding errors, it is possible that the algorithm will not identify the right subspace, but only an approximation to it.

\vspace{10pt}
\noindent \textbf{The deterministic method}

We choose a simple distribution $f(X,\psi) dX d\psi$, say, a Gaussian, for the entries of each of the matrices $X_1,...,X_n$ and the components of the unnormalized vector $\psi$. Then we define the components of the $2k$ moment vector $y(X,\psi)$ via the relation $y(X,\psi)_w\equiv \bra{\psi}w(X)\ket{\psi}$. Then we compute analytically the matrix

\be
S\equiv\int f(X,\psi)dXd\psi y(X,\psi)y(X,\psi)^\dagger.
\ee

\noindent Clearly, the space $\S^k_D$ corresponds to the support of $S$. Diagonalizing $S$ and keeping just the eigenvectors with non-zero eigenvalue we hence obtain an orthonormal basis for $\S^k_D$.

The disadvantage of this method is that it involves symbolic computations, and hence, depending on the platform used, it is either more difficult to code or results in slower programs.

\section{Convergence of the hierarchy}
\label{convergence}

Let $p^k$ denote the result of the $k^{th}$ order relaxation~(\ref{npo_rel}) of problem~(\ref{npo_dim}), and let $y^k$ be the corresponding minimizer $2k^{th}$-order moment vector.

Now, let us assume that the Archimedean condition~(\ref{archimedean}) is met, and call $r$ the degree of the polynomial on the right-hand side of eq.~(\ref{archimedean}). Expressing the polynomials $f_i,g_{ij}$ as $f_i(X)=\sum_v f^v_i v(X)$, $g_{ij}(X)=\sum_v g^v_{ij} v(X)$, it follows that

\begin{align}
&Cy^k_{u^\dagger u}-\sum_{l=1}^ny^k_{u^\dagger X_l^2u}=\sum_{i}\sum_{v,w}(f_i^v)^*f_i^wy^k_{u^\dagger v^\dagger w u}+ \nonumber\\
&+\sum_{i,j}\sum_{v,w,s}(g_{ij}^v)^*g_{ij}^wq^s_iy^k_{u^\dagger v^\dagger s w u},
\end{align}

\noindent for $|u|\leq k-\lceil\frac{r}{2}\rceil$. Due to positive semidefiniteness of the moment and localizing matrices, the right hand side of the above equation is non-negative, implying that $Cy^k_{u^\dagger u}\geq y^k_{u^\dagger X_l^2u}$, for all $X_l$. By induction, it follows that

\be
C^{|u|}\geq y^k_{u^\dagger u}
\ee

\noindent for all sequences $|u|\leq k-\lceil\frac{r}{2}\rceil$. Such moments corresponds to the diagonal entries of the moment matrix $M_k(y^k)$. Since $M_k(y^k)\geq 0$, it follows that $|y_{|w|}|\leq C^{|w|/2}$ for all monomials $|w|\leq 2k-2\lceil\frac{r}{2}\rceil$.

Now, for each vector $y^k$ replace with zeros all entries $y^k_w$, with $|w|> 2k-2\lceil\frac{r}{2}\rceil$, and complete the resulting $2k$-order moment vector to an $\infty$-order moment vector by placing even more zeros. We arrive at an infinite sequence $\hat{y}^s,\hat{y}^{s+1},...$ of vectors, with $|\hat{y}^k_w|\leq C^{|w|}$ for all $k,w$. By the Banach-Alaoglu theorem~\cite{reed_simon}, this sequence has a converging subsequence~\footnote{Technically, the Banach-Alaoglu theorem must be applied to the sequence $\hat{z}^s,\hat{z}^{s+1},...$, where $z_u=\frac{y_u}{C^{|u|/2}}$.}, and we will call $\hat{y}$ the corresponding limit.

$\hat{y}$ satisfies $\hat{y}_\id=1$, and $M_k(\hat{y}),M_k(q_i\hat{y})\geq 0$ for all $k,i$. By successive Cholesky decompositions of $M_k(\hat{y})$, for $k=s,s+1,...$, we find a sequence of complex vectors $(\ket{u})_u$ with the property $\hat{y}_{u^\dagger v}=\braket{u}{v}$, for all monomials $u,v$. Call $H\equiv \mbox{span}\{\ket{u}:u\}$. We define the action of the operator $\tilde{X}_i$ on this (non-orthogonal) basis by

\be
\tilde{X}_i\ket{u}=\ket{X_iu}.
\ee

\noindent and extend its definition to $\mbox{span}\{\ket{u}:u\}$ by linearity. To prove that this definition is consistent, we need to show that, if $\sum_{u}c_u\ket{u}=\sum_u d_u\ket{u}$ for two different linear combinations $(c_u)_u,(d_u)_u$, then $\sum_{u}c_u\ket{X_iu}=\sum_u d_u\ket{X_iu}$. Indeed, note that, for any vector $\ket{w}$,

\begin{align}
&\bra{w}\sum_{u} c_u\ket{X_iu}=\sum_{u} c_u\braket{w}{X_iu}=\sum_{u} c_u\braket{X_iw}{u}=\nonumber\\
&=\bra{X_iw}\sum_{u} c_u\ket{u}=\bra{X_iw}\sum_{u} d_u\ket{u}=\bra{w}\sum_{u} d_u\ket{X_iu},
\end{align}

\noindent where the $2^{\text{nd}}$ and $5^{\text{th}}$ equalities follow from $\braket{w}{X_iu}=y_{w^\dagger X_i u}=y_{(X_iw)^\dagger u}=\braket{X_iw}{u}$. This relation holds for arbitrary $\ket{w}$, so the vectors $\sum_{u}c_u\ket{X_iu}$, $\sum_u d_u\ket{X_iu}$ must be identical. Similarly, it can be verified immediately that $\tilde{X}_i$ is a symmetric operator, since $\bra{u}\tilde{X}_i\ket{v}=y_{u^\dagger X_i v}=y_{v^\dagger X_i u}^*=\bra{u}\tilde{X}_i\ket{v}$

From the positive semidefiniteness of the localizing matrices $M(q_i\hat{y})$, it can be shown that $\bra{\phi}q_i(\tilde{X})\ket{\phi}\geq 0$ for all $\ket{\phi}\in H$ and $i=1,...,m$. The Archimedean condition implies, moreover, that $\bra{\phi}\tilde{X}^2_i\ket{\phi}\leq C\braket{\phi}{\phi}$. From this observation it is trivial to extend the action of $\tilde{X}_i$ to $\tilde{\H}$, the closure of $H$, and hence we arrive at a Hilbert space $\tilde{\H}$ and a set of operators $\tilde{X}_1,...,\tilde{X}_n$ such that $q_i(\tilde{X})\geq 0$ for $i=1,...,m$ and $\hat{y}_u=\bra{\tilde{\psi}}u(\tilde{X})\ket{\tilde{\psi}}$, for $\ket{\tilde{\psi}}\equiv\ket{\id}$.

Note as well that, by construction, these operators satisfy all MPIs for dimension $D$. Now, call $\A$ the von Neumann algebra generated by $\tilde{X}_1,...,\tilde{X}_n$. By von Neumann's 1949 result~\cite{von_neumann}, such an algebra must decompose as a direct integral of type I, II and III factors~\cite{takesaki}. That is,

\be
\A=\int^{\oplus}d\mu^\text{I} (y) \A^\text{I}_y\oplus \int^{\oplus}d\mu^\text{II} (y) \A^\text{II}_y\oplus \int^{\oplus}d\mu^\text{III}(y) \A^\text{III}_y.
\label{decomp}
\ee

Type I factors are isomorphic to $B(\H)$, for Hilbert spaces $\H$ of finite or infinite dimensionality~\cite{takesaki}. Since $\A$ must satisfy the MPIs for dimension $D$, that excludes Hilbert spaces of dimension $d>D$ from the first term of the right-hand side of~(\ref{decomp}). Moreover, in the Appendix it is proven that Type II and III factors violate the standard identity~(\ref{funda}) for all values of $d$.

It follows that we can write our operators $\tilde{X}_1,...,\tilde{X}_n$ as

\be
\tilde{X}_i=\int^{\oplus}d\mu^\text{I} (y) \tilde{X}_{i,y},
\ee

\noindent where each $\tilde{X}_{i,y}$ acts on a Hilbert space $\H_y$ with $\mbox{dim}(\H_y)\leq D$. From $q_i(\tilde{X})\geq 0$, it follows that $q_i(\tilde{X}_{y})\geq 0$ for $i=1,...,m$. Hence $\hat{p}$ is a convex combination of feasible values of $\langle p(X)\rangle$ and so $\hat{p}\leq p^\star$. On the other hand, $p^k\geq p^\star$ for $k\geq s$. Thus $\hat{p}=\lim_{k\to\infty}p^k\geq p^\star$, proving the convergence of the hierarchy.

%Using the GNS construction as in \cite{siam}, we can hence establish the existence of a Hilbert space $\H$, a normalized vector $|ket{\tilde{\psi}}\in\H$ and self-adjoint operators $\tilde{X}_1,...,\tilde{X}_n\subset B(\H)$ such that

\begin{remark}

Note that we just invoked the Archimedean condition~(\ref{archimedean}) to establish the existence of $\hat{y}$ and, later, the boundedness of the operators $\tilde{X}_1,...,\tilde{X}_n$. Both results also follow from the weaker \emph{D-dimensional Archimedean condition}:

\begin{multline}
C-\sum_{i}X_i^2=\sum_{j}f_j(X)^\dagger f_j(X)\\+\sum_{i,j}g_{ij}(X)^\dagger q_i(X)g_{ij}(X)+h_D(X),
\label{archimedean2}
\end{multline}

\noindent where $h_D(X)$ is an MPI for dimension $D$.

\end{remark}

\begin{remark}

If we take $D=1$, then the MPIs will force all operators $X_1,...,X_n$ to commute with each other. In that case, the SDP hierarchy reduces to the Lasserre-Parrilo hierarchy for polynomial minimization~\cite{lasserre,parrilo}.

\end{remark}

\begin{remark}

So far, we have been assuming that the variables $X_1,...,X_n$ are Hermitian. If a subset of them is not, one can still define a converging SDP hierarchy, by considering all possible monomials of the variables $X_1,...,X_n$ and their adjoints $X^\dagger_1,...,X^\dagger_n$ in the definition of moment vectors and moment matrices. In that case, the left-hand side of the $D$-dimensional Archimedean condition~(\ref{archimedean2}) must be replaced by $C-\sum_{i}X_iX^\dagger_i-X^\dagger_iX_i$.

\end{remark}

\section{Exploiting polynomial constraints}
\label{exploit}

It is a basic result in operator algebras that MPIs for $D\times D$ matrices must have degree at least $2D$~\cite{MPI}. This implies that we would need to implement the $D^\text{th}$ relaxation of~(\ref{npo_rel}) or~(\ref{npo_rel2}) in order to obtain non-trivial $D$-dimensional constraints. Even for problems involving a small number of noncommuting variables, this becomes impractical already for $D=5$. Hence, if we wish to conduct optimizations over matrices of dimensions greater than 2 or 3, we must rely on linear restrictions other than those derived from MPIs.

Most NPO problems relevant in quantum information science involve polynomial \emph{identities} rather than polynomial \emph{inequalities}. That is, constraints of the form $q(X)\geq 0$ are complemented with $-q(X)\geq 0$, and so $q(X)=0$ must hold for all representations of $\{q_i(X)\geq 0\}$. The strategy we will follow to solve this kind of problems is to divide the representations of $X_1,...,X_n$ into different classes $r$, in such a way that any two representations belonging to the same class $r$ can be connected by a continuous trajectory of feasible representations in $r$. As we will see, each of these classes will satisfy non-trivial low degree polynomial identities, which we will translate into linear constraints at the level of moment matrices (vectors). By carrying out a relaxation of the form~(\ref{npo_rel}) for each possible class $r$ and taking the greatest result we hence obtain an upper bound on the solution of the general problem~(\ref{npo_dim}).

For instance, suppose that $\{q_i(X)\geq 0\}$ contains relations of the form $X_i^2=1$, for $i=1,2,3$. For $D=2$, there are two possibilities:

\begin{enumerate}
\item
$X_i=\pm \id$ for some $i\in\{1,2,3\}$.
\item
$X_i\not=\pm \id$ for all $i$, in which case it can be shown that the operators satisfy the identities:

\begin{align}
&[X_1,\{X_2,X_3\}_{+}]=[X_2,\{X_1,X_3\}_{+}]=\nonumber\\
&=[X_3,\{X_1,X_2\}_{+}]=0.
\end{align}
\end{enumerate}

\noindent In either case, the noncommuting variables satisfy non-trivial polynomial constraints of degree smaller than 4, the smallest possible degree of an MPI for $D=2$. A way to attack this problem is therefore to define an SDP relaxation of the form~(\ref{npo_rel2}) for each case, enforcing the corresponding extra linear constraints on the moment matrix (or moment vector).

Note also that, if we further assume that the matrices $X_1,...,X_2$ are real, then we can add constraints of the sort

\be
\{X_1,[X_2,X_3]\}_{+}=0
\ee

\noindent in the second case. This approach hence allows us (in principle) to distinguish between real and complex matrix algebras.

The objective, again, is to identify all possible linear restrictions on $\hat{M}_k$ or $y$ within a given class $r$. Fortunately, most NPO problems in quantum information science have the peculiarity that random representations of a given class can be generated efficiently.

Continuing with the previous example, suppose that we wish to optimize over $6$ dichotomic operators, i.e., the polynomial constraints are, precisely, $X^2_i=\id$, for $i=1,2,...,6$. For simplicity, let us denote the first four operators as $X_{00},X_{01},X_{10},X_{11}$ and the last two by $Y_1,Y_2$. We want to maximize the average value of the operator

\be
p(X)\equiv\sum_{j=1,2}\sum_{c_1,c_2=0,1}(-1)^{c_j}X_{c_1c_2}Y_{j}X_{c_1,c_2}.
\label{obj}
\ee

\noindent Note that we can write each dichotomic operator as $X_i=(-1)^a\frac{2E^{i}_a-\id}{2}$, where $\{E^i_a\}$ are projection operators satisfying $E^i_0+E^i_1=\id$. Substituting in~(\ref{obj}), we have that the objective function $\bra{\psi}p(X)\ket{\psi}$ is equal to

\be
4\sum_{c_1,c_2,s=0,1}P(s,c_1|X_{c_1c_2},Y_1)+P(s,c_2|X_{c_1c_2},Y_2)-16,
\ee

\noindent where $P(a_1,a_2|x_1,x_2)=\bra{\psi}E^{x_1}_{a_1}E^{x_2}_{a_2}E^{x_1}_{a_1}\ket{\psi}$. This corresponds to the \emph{temporal correlations scenario} defined in~\cite{temp_corr}, where sequential dichotomic projective measurements are conducted over a quantum system, and a record of the measurements $x_1,x_2,...$ implemented, as well as the measurement outcomes $a_1,a_2,...$, is kept. The goal is to limit the statistics $P(a_1,...,a_n|x_1,...,x_n)$ obtained after several repetitions of the experiment.

As noted in~\cite{siam,temp_corr}, when the dimensionality of the quantum system is unrestricted, the set of all feasible distributions $P(a_1,...,a_n|x_1,...,x_n)$, and hence the optimal value of~(\ref{obj}), can be characterized by a single SDP. Using the SDP solver Mosek~\cite{mosek}, we find that, for $D=\infty$, $p^\star= 8$ up to 7 decimal places.

Suppose, however, that we have the promise that the system has dimension $D=2$. The problem we want to solve is therefore

\begin{align}
p^\star=&\max_{\H,X,\psi}\bra{\psi}p(X)\ket{\psi}\nonumber\\
\mbox{s.t. } &\dim(\H)\leq 2, \id-X_i^2=0,\mbox{ for } i=1,...,6.
\label{npo_dich}
\end{align}

We start by dividing the representations of $2$-dimensional dichotomic operators into classes. For any dichotomic operator, the rank of the projector $E\equiv\frac{X+\id}{2}$ can be $0$, $1$ or $2$. For $r=0,2$, the corresponding operator is $X=-\id$ or $X=\id$, respectively. For $r=1$, a random dichotomic operator $X$ can be generated as $X=2\frac{\proj{v}}{\braket{v}{v}}-\id$, where $v\in\C^2$ is a random complex vector. Since we are dealing with six non-commuting variables, there are $3^6=729$ classes, labeled by the vector $\vec{r}\in\{0,1,2\}^6$, with $\mbox{rank}(X_i+\id)=r_i$.

For a fixed value of $\vec{r}$, we sequentially generate random $6$-tuples of dichotomic operators $X^j\equiv(X^j_1,...,X_6^j)$, with the required rank constraints, as well as a sequence of random normalized vectors $\ket{\psi^j}\in \C^2$. As before, we use each pair $(X^j,\ket{\psi^j})$ to generate a random feasible moment matrix $M^j_k$. Note that, since the conditions $X_i^2=\id$ are implicit in each moment matrix, it is not necessary to include localizing matrices in our description (they would amount to zero diagonal blocks in the extended moment matrix). Notice as well that, given a feasible moment matrix $M_k$, its complex conjugate $M_k^*$ is also feasible. Since $p(X)$ is a real linear combination of Hermitian monomials, the objective function will have the same value for both $M_k$ and $M_k^*$ (and thus for the real feasible moment matrix $\mbox{Re}(M_k)=\frac{M_k}{2}+\frac{M_k^*}{2}$). This implies that, in order to define an SDP relaxation for~(\ref{npo_dich}), it suffices to consider the sequence of real matrices $\mbox{Re}(M_k^1),\mbox{Re}(M_k^2),...$.

Applying the modified Gram-Schmidt method to that sequence until we find linear dependence, we obtain an orthonormal basis for ${\cal M}^k_{D,\vec{r}}$, the space of all real feasible moment matrices for representations of the class $\vec{r}$. This time, the fact that this randomizing method works with probability one is a consequence that the projection of the randomly generated moment matrix $\mbox{Re}(M_k^{j+1})$ onto the orthogonal complement of the space spanned by $\mbox{Re}(M_k^{1}),...,\mbox{Re}(M_k^{j})$ is a matrix whose entries are rational functions of the randomly generated vectors used to build $X^{j+1}$ and $\ket{\psi^{j+1}}$. If $\mbox{Re}(M_k^{1}),...,\mbox{Re}(M_k^{j})$ do not span ${\cal M}^k_{D,\vec{r}}$, then there exists a choice for those vectors such that the projected matrix is non-zero, i.e., at least one of such rational functions is non-zero. It is a well-known fact that the probability that a randomly evaluated non-zero rational function vanishes is zero.

Alternatively, we can identify ${\cal S}^k_{D,\vec{r}}$, the space of feasible $2k$-order moment vectors for representations in the class $\vec{r}$ by parametrizing each normalized vector needed to build $X$ or $\psi$ by two angles $\phi,\varphi$, and constructing the corresponding moment vector $y(\vec{\phi},\vec{\varphi})$. Then the entries of the matrix

\be
S\equiv \int d\vec{\phi} d\vec{\varphi} \mbox{Re}(y(\vec{\phi},\vec{\varphi})) \mbox{Re}(y(\vec{\phi},\vec{\varphi})^\dagger)
\ee

\noindent can be computed analytically. Its support will coincide with ${\cal S}^k_{D,\vec{r}}$.

One way or the other, we must solve the program

\begin{align}
p^k=&\max_{\hat{M}}\sum_wp_w \hat{M}_{w,\id},\nonumber\\
\mbox{s.t. } &\hat{M}\in \M^k_{D\vec{r}}, \hat{M}_{\id,\id}=1,\hat{M}\geq 0.
\label{npo_dich_rel}
\end{align}

\noindent for all possible classes $\vec{r}$. For $k=2$, again using Mosek~\cite{mosek}, we find $p^2\approx 5.656854$, definitely smaller than the free limit.

\section{Similarly inspired SDP hierarchies}
\label{others}

In the following we will introduce two problems in quantum information science which, while not exactly fitting in the class of problems~(\ref{npo_dim}), can be similarly reduced to SDP hierarchies.

\subsection{Quantum nonlocality under dimension constraints}

The scenario is as follows: two distant parties, call them Alice and Bob, conduct measurements on a bipartite quantum system. We denote Alice's (Bob's) measurement setting by $x$ ($y$) and her (his) measurement outcome by $a$ ($b$). We wish to bound a linear functional of the statistics $P(a,b|x,y)$ they will observe, under the assumption that Alice's and Bob's spaces are, at most, $D$-dimensional. If Alice and Bob's outcomes are binary, i.e., $a,b\in\{0,1\}$ the problem can be shown equivalent to

\bea
&&\max \sum_{x,y,a,b}B^{x,y}_{a,b}P(a,b|x,y),\nonumber\\
\text{s.t.} && P(a,b|x,y)=\bra{\psi}E^x_a\otimes F^y_b\ket{\psi},
\label{Bell_dim}
\eea

\noindent where $\{E^x_a,F^y_b\}$ are projection operators acting on $\C^D$, with $\sum_aE^x_a=\sum_bF^y_b=\id_D$, and $\ket{\psi}\in \C^{D}\otimes\C^D$.

Following the last section, we divide the representations of the operators $E^x_a,F^y_b$ into different classes labeled by the vectors $\vec{r},\vec{t}$, with $\mbox{rank}(E^x_a)=r^x_a$, $\mbox{rank}(F^y_b)=t^y_b$. For each representation class $\vec{r},\vec{t}$, we try to characterize the span $\M^k_{D,\vec{r}\vec{t}}$ of feasible $k^{th}$-order moment matrices. To do so, we sequentially generate random normalized states $\ket{\psi^j}\in\C^D\otimes\C^D$ and projectors $E^{x,j}_a,F^{y,j}_b\in B(\C^D)$, satisfying the rank conditions $\mbox{rank}(E^{x,j}_a)=r^x_a$, $\mbox{rank}(F^{y,j}_b)=t^y_b$. Given $E^{x,j}_a,F^{y,j}_b$, we define the projectors $\bar{E}^{x,j}_a\equiv E^{x,j}_a\otimes \id_D$ and $\bar{F}^{y,j}_b\equiv \id_D\otimes F^{y,j}_b$, which we will use to generate a feasible $k^{th}$-order moment matrix $M_k^j$. By subjecting the resulting sequence of moment matrices to the modified Gram-Schmidt orthogonalization, we obtain a basis for $\M^k_{D,\vec{r}\vec{t}}$. The SDP to solve is hence

\bea
B^k\equiv&&\max \sum_{x,y,a,b}B^{x,y}_{a,b}(M_k)_{\bar{E}^{x}_a,\bar{F}^{y}_b}\nonumber\\
\text{s.t.} && (M_k)_{\id,\id}=1, M_k\geq 0,M_k\in{\cal T}^k_{D,\vec{r},\vec{t}}.
\label{Bell_dim_rel}
\eea

\noindent We again advise the reader to check that $T\equiv \frac{1}{N}\sum_{j=1}^NM_k^j$ is positive definite: otherwise, a projection of $M_k$ onto the support of $T$ is necessary to guarantee the strict feasibility of the associated SDP, see Remark~\ref{support}.

The completeness of the above SDP hierarchy can be established easily: following the same lines as in Section~\ref{convergence}, we prove the existence of a (in general, infinite-dimensional) representation $\tilde{E}^x_b,\tilde{F}^x_b\subset B(\H)$, with $[\tilde{E}^x_b,\tilde{F}^x_b]=0$ for all $x,y,a,b$, and a state $\tilde{\psi}$, such that $\sum_{x,y,a,b}B^{x,y}_{a,b}\bra{\tilde{\psi}}\tilde{E}^{x}_a\tilde{F}^{y}_b\ket{\tilde{\psi}}$ coincides with the asymptotic limit $\hat{B}\equiv\lim_{k\to \infty}B^k$. The center of the algebra $\A$ generated by $\{\tilde{E}^x_a:x,a\}$ decomposes $\H$ into a direct integral of sectors $\H_z$. By construction, in each sector $z$, $\A$ boils down to a type I factor $\A_z$ of dimension smaller than or equal to $D$. Being a type-I factor, we can write $\H_z=\H_z^A\otimes\H_z^B$; then $\A_z\sim B(\H_z^A)\otimes \id$, and $\A'_z\sim \id\otimes B(\H_z^B)$, where $\A'_z$ denotes the commutant of $\A_z$. It follows that

\begin{align}
&\tilde{E}_a^x=\int^{\oplus}d\mu(z) \tilde{E}^x_{a,z}\otimes\id_{B,z},\nonumber\\
&\tilde{F}_a^x=\int^{\oplus}d\mu(z) \id_{A,z}\otimes\tilde{F}^y_{b,z},
\end{align}

\noindent where $\mbox{dim}(\H_{A,Y})\leq D$. Likewise, it can also be shown that the algebra generated by $\tilde{F}^y_{b,z}$ decomposes as a direct integral of finite dimensional algebras with dimension smaller than or equal to $D$. $\hat{B}$ is thus a convex combination of feasible points and as such it represents a lower bound for the original problem~(\ref{Bell_dim}).

\vspace{10pt}
\noindent\textbf{Examples}

Here we give examples of maximizing the violation of bipartite Bell inequalities with binary outcomes for different dimensionality of the component spaces. Some of the examples have already appeared in Ref.~\cite{finite_dim_short}. First we discuss the $I_{3322}$ inequality, the only tight three-setting two-outcome Bell inequality and its modified version. Then we move on to one more settings per party. For all the subsequent computations we used the solvers Mosek~\cite{mosek} and SeDuMi~\cite{sedumi} through the interface YALMIP~\cite{yalmip}, which we ran on a memory enhanced desktop PC (with 128 GB RAM).

\noindent\emph{I3322} --- First we considered the $I_{3322}$ inequality~\cite{i3322}, which is the member of the $I_{NN22}$ family $N\ge 2$. Recently, it has been proven that qubit systems are not enough to attain the overall quantum maximum $0.2509$. Rather, the best value in $\C^2\times\C^2$ systems is 0.25~\cite{head_leg,moroder}.
Using SDP, we reproduced the maximum value of 0.25 in dimensions $\C^3\times\C^3$ as well up to 8 significant digits~\cite{finite_dim_short}. The size of the moment matrix was 76 involving 1240 linear constraints. The computations took about 5 minutes for a fixed rank combination of measurements. Note that the hierarchy of Moroder et al.~\cite{moroder}, by limiting the negativity~\cite{negativity} of the bipartite quantum state, also gives a (not necessarily tight) upper bound on the Bell violation for a fixed dimension of the quantum state. Indeed, this method works for $\C^2\times\C^2$ systems by returning a violation of 0.25 of the $I_{3322}$ inequality~\cite{moroder}. However, for $\C^3\times\C^3$ systems it does not seem to converge (see Fig.~1 of~\cite{moroder}).

The SDP method also allows the user to upperbound the maximum quantum violation of a Bell inequality using a fixed two-qudit state. Let us choose the 3-dimensional maximally entangled state, $\ket{\psi} = (\ket{00}+\ket{11}+\ket{22})/\sqrt 3$. We find the value of $0.229771$, which is saturated by see-saw computation, hence the presented upper bound is tight. The computation involved a 116 dimensional moment matrix (on a partial 4-level relaxation) along with 1060 linear constraints. The program took 2 minutes to complete for a fixed rank combination of projective measurements. We mention a related problem, where the maximal violation of $I_{3322}$ has been computed for the maximally entangled state (of unrestricted dimensionality). This has been solved both analytically~\cite{vidick} and using a relaxation method~\cite{zoo} by returning the value of 0.25.

\noindent\emph{Modified I3322} --- Though, $I_{3322}$ inequality likely requires infinite dimensions to achieve the maximal violation $0.2509$~\cite{npa}, $D=12$ seems to be the smallest local dimension surpassing the qubit bound 0.25~\cite{palvert}. It is an open question whether there exists a Bell inequality for which the maximum quantum violation in a given dimension is a strictly monotonic function of the dimension. Below we give such a candidate. To this end, we modify the $I_{3322}$ inequality by introducing a parameter $c\ge1$:
\begin{align}
\label{i3322c}
I_{3322}(c)=& E^A_1 + E^A_2 + E^B_1 + E^B_2\nonumber\\
&-(E_{1,1}+E_{1,2}+E_{2,1}+E_{2,2}) \nonumber\\
&+c(E_{1,3}+E_{3,1}-E_{2,3}-E_{3,2})\le 4c,
\end{align}
where the correlator $E_{x,y}$ between measurement $x$ by Alice and measurement $y$ by Bob is defined as $E_{x,y}=P(a=b|x,y)-P(a\neq b|x,y)$, $a,b\in\{0,1\}$ and $E^A_x$ denotes the marginal of Alice's measurement setting $x$ (and $E^B_y$ is similarly defined for Bob). This inequality is symmetric for exchange of Alice and Bob, and returns the original $I_{3322}$ inequality (written in terms of correlators) for parameter $c=1$.

Setting $c=2$, we used the see-saw variational technique~\cite{seesaw1,seesaw2} to find a lower bound on the maximal violation for any dimension $2\le d\le 15$, which we observe to be gradually increasing with dimension. We conjecture that the bounds are tight. Table~I shows results up to $D=6$ concerning both the lower (see-saw) and upper bounds (SDP). Accordingly, the bounds for $D=2,3$ are indeed tight. Computationally the most challenging case was obtaining the upper bound in $D=4$. It involved 3514 constraints, the dimension of the moment matrix is 184 and took roughly 40 minutes for Mosek to complete the task for a given rank combination of the measurements. The quantum maximum in dimensions 5 and 6 coincide with the NPA bound on level 3 up to the shown digits. We pose it as a challenge to prove tightness of the see-saw bound for $D=4$ (or possibly higher dimensions) by exploiting the symmetric structure of the inequality~(\ref{i3322c}) using techniques such as in Refs.~\cite{moroder,faacets}.

\begin{table}
\label{i3322trunc}
\begin{center}\begin{tabular}{ccc}
\hline
  $D$ & Lower bound & Upper bound \\
  \hline
  2 & 8.013177 & 8.013177\\
  3 & 8.024050 & 8.024050\\
  4 & 8.032766 & 8.071722\\
  5 & 8.039579 & 8.075937\\
  6 & 8.056714 & 8.075937\\
  \hline
\end{tabular}
\caption{Quantum bounds for different local dimensions on the violation of the
$I_{3322}(2)$ inequality computed using see-saw search/SDP computation. Bounds for $D=2,3$ are tight, since the see-saw and SDP bounds match. As an overall upper bound, the NPA hierarchy on level 3 gives 8.075937.}
\end{center}
\end{table}

\noindent\emph{I4422 family} --- A one-parameter family of four-setting inequalities is given in Ref.~\cite{head_leg}. These inequalities are not tight but they have a quite simple structure. They look as follows for $c\ge 0$:
\begin{align}
I_{4422}(c)=& c E^A_1 + (E_{1,1}+E_{1,2}+E_{2,1}-E_{2,2}) \nonumber\\
&+(E_{3,3}+E_{3,4}+E_{4,3}-E_{4,4})\le 4+c,
\end{align}
When $c=0$, it is a direct sum of two CHSH inequalities, hence maximum violation is attained with qubit systems. However, by setting $c>0$, it may serve as a dimension witness. In particular, for $c=1$ (the value used in Eq.~(19) of~\cite{head_leg}), its maximal violation in $\C^2\times\C^2$ systems is upper bounded by the value of 5.8515~\cite{head_leg}. However, using our SDP tool, this upper bound turns out to be not tight: We certify a smaller value of 5.8310, which is matched by the see-saw method. Further, by raising the dimension to $\C^3\times\C^3$, we get the same amount of violation. The SDP computation returning 5.8310 in $\C^3\times\C^3$ was quite demanding: it required a 130 dimensional moment matrix and took about 2 hours of computational time. 5.8310 must be compared to the maximum value of $2\sqrt 2 + \sqrt{10}\approx 5.9907$, achievable in $\C^4\times\C^4$ systems. In contrast to our certified value 5.8310, the corresponding $\C^3\times\C^3$ value arising from Moroder et al. hierarchy~\cite{moroder} (on their level 2) is a higher value of 5.9045.

\noindent\emph{I4722 inequality} --- It is also worth mentioning a situation (actually, this is the only case we are aware of) for which a previous SDP method introduced in Ref.~\cite{head_leg} outperforms our present SDP method. We tested the method in case of asymmetric Bell inequalities, that is, when the number of settings on the two sides are not the same. For the sake of comparison, we have chosen a correlation-type Bell inequality from~\cite{VP09}, already analyzed in~\cite{head_leg}:

\begin{align}
I_{4722}&= E_{11}+E_{21}+E_{31}+E_{41}\nonumber\\
& +(E_{12}-E{22}) + (E_{31}-E{33}) + (E_{41}-E{44})\nonumber\\
& +(E_{25}-E{35}) + (E_{26}-E{46}) + (E_{37}-E{47})\le 8,
\end{align}
which consists of 4 and 7 binary-outcome settings on Alice and Bob's respective sides. In Ref.~\cite{head_leg} a method is presented in Sec. IIIB, which is particularly suited to asymmetric Bell setups. This way, the best upper bound obtained for qubit systems is 10.5102, whereas the best lower bound value of 10.4995 is due to see-saw search. The quantum maximum, attainable with two ququarts, is 10.5830. Using our present SDP technique and a desktop PC, unfortunately, we did not manage to go below the global maximum 10.5830.

%\textbf{Bounds for the violation of the Bell inequality defined in [Physical Review A, vol 83, art. 022108 (2011)] by maximally entangled states, POVM witnesses, etc.}

\vspace{10pt}

Suppose now that Alice and Bob are conducting non-binary measurements, that is, measurements with more than just two outcomes. Then, in order to consider the most general measurements they could perform, we must model their measurement devices via Positive-Operator Valued Measures (POVMs), rather than projective measurements. There are two ways to accomplish this:

\begin{enumerate}
\item
We can replace constraints of the sort $(E^x_a)^2=E^x_a$ in~(\ref{Bell_dim}) by the positive semidefinite constraints $E^x_a\geq 0$, at the cost of having to add the corresponding localizing matrices to~(\ref{Bell_dim_rel}). Although converging, this method does not seem to behave well in our numerical experiments.

\item
Alternatively, we can exploit the fact that any $d$-outcome POVM $\{E_a\geq 0\}\subset B(\C^D)$ can be realized in an extended Hilbert space $\C^d\otimes\C^D$ via a projective measurement of the form $M_a =U(\proj{a}\otimes \id_D)U^\dagger$, where $U\in B(\C^d\otimes\C^D)$ is a unitary matrix~\cite{nielsen_chuang}. Indeed, taking the state to be $\rho=\proj{0}\otimes \proj{\psi}$ and choosing $U$ appropriately, it can be verified that

\be
\tr(\rho M_a)=\tr(E_a\proj{\psi}),
\ee

\noindent for $a=0,...,d-1$ and all states $\ket{\psi}$.

In the hierarchy to implement, random states of the form $\proj{0}_{A'}\otimes \proj{\psi}_{AB}\otimes\proj{0}_{B'}$ are generated. For each random state, we construct a moment matrix containing the operators $\bar{E}^x_a=U^x(\proj{a}\otimes \id_D)(U^x)^\dagger\otimes\id_D\otimes\id_d$, $\bar{F}^y_b=\id_d\otimes\id_D\otimes V^y(\id_D\otimes \proj{b})(V^y)^\dagger$ and the projectors $P_A=\proj{0}\otimes \id^{\otimes 2}_D\otimes\id_d$, $P_B=\id_d\otimes \id^{\otimes 2}_D\otimes\proj{0}$.

The convergence of this hierarchy follows from the fact that the algebras generated by $\{P_A \bar{E}^x_aP_A\}$, $\{P_B \bar{F}^y_bP_B\}$ cannot violate $D$-dimensional MPIs.

\end{enumerate}

\vspace{10pt}
\noindent\textbf{Examples}

We now apply our method to place nontrivial upper bounds on the quantum violation of Bell inequalities using genuine POVM measurements for some of the settings. Note for binary-outcome settings general POVM measurements are not relevant, hence we have to consider Bell inequalities with at least one non-binary setting. To this end, we consider the simplest tight Bell inequality due to Pironio beyond genuine two-outcome inequalities~\cite{d_wit1,pironio}. In this inequality, Alice has three binary-outcome measurements, and Bob has two settings: the first one has binary outcomes and the second one has ternary outcomes. If we allow Bob to use general POVM measurements on his second setting, the two-qubit quantum maximum $(\sqrt 2 - 1)/2\approx 0.2071$ is recovered up to computer precision on level 3 of the SDP hierarchy. Hence, in this particular Bell inequality the use of general measurements do not provide any advantage over projective ones. Let us note that the quantum maximum without dimension constraints is a larger value 0.2532 which can be obtained using a two-qutrit system and projective measurements~\cite{d_wit1}. We also applied the above method to the CGLMP inequality in order to prove the conjecture that the qubit bound $0.2071$ using projective measurements is optimal (i.e. general POVM measurements do not help to improve the bound). However, in that case, we were unable to go below the known overall quantum maximum given in Refs.~\cite{latorre,npa}.

The previous approach can be easily extended to characterize the statistics of multipartite scenarios where the local dimensionality of all parties is bounded from above. More interestingly, it can also be adapted to deal with multipartite Bell scenarios where only a subset of the parties has limited dimensionality.

Consider, for instance, a tripartite scenario where Alice and Bob's measurement devices are unconstrained, but the dimensionality of the third system (say, Charlie's) is bounded by $D$. We want to generate a basis for the corresponding space of truncated moment matrices, with rows and columns labeled by strings of operators of the form $u(AB)v(C)$, where $u(AB)$ ($v(C)$) denotes a string of Alice and Bob's (Charlie's) operators of length at most $k_{AB}$ ($k_C$).

The key is to realize that, in a multipartite (complex) Hilbert space, the space of feasible moment matrices is spanned by moment matrices corresponding to separable states. Hence, in order to attack this problem, we start by generating a sequence of \emph{complex} $D$-dimensional moment matrices for Charlie's system alone. After applying Gram-Schmidt to these complex matrices, we obtain the basis of Hermitian matrices $\{M_j\}_{j=1}^N$. Next, we generate a basis for Alice and Bob's moment matrices. Since their dimension is unconstrained, such matrices are expressed as

\be
\Gamma_{k_{AB}}=\sum_{|u|\leq 2k_{AB}}c_uN_u+c^*_uN_{u^\dagger},
\label{span_inf}
\ee

\noindent where $N_u$ is a matrix defined by

\bea
(N_u)_{v,w}=&1, \mbox{ if } v^\dagger w=u;\nonumber\\
&0, \mbox{ otherwise}.
\eea

\noindent The overall moment matrix for the whole system can then be expressed as $M=\sum_{u,j}M_j\otimes (c_{u,j}N_u+c_{u,j}^*N_{u^\dagger})$.

Since we are just interested in optimizing a real linear combination of real entries of $M$---corresponding to the measured probabilities $P(a,b,c|x,y,z)$---, we can take the real part of the above matrix, and so we end up with the relaxation:

\begin{gather}
\max \sum_{x,y,z,a,b,c}B^{x,y,z}_{a,b,c}M_{E^x_a,F^y_bG^z_c},\nonumber\\
\mbox{s.t. } M_{\id,\id}=1,M\geq 0,\nonumber\\
M=\sum_{u,j}c^{\R}_{u,j} \mbox{Re}(M_j)\otimes (N_u+N_{u^\dagger})-\nonumber\\
-c^{\mathbb{I}}_{u,j}\mbox{Im}(M_j)\otimes (N_u-N_{u^\dagger}),
\end{gather}

\noindent where $c^{\R}_{u,j}$ ($c^{\mathbb{I}}_{u,j}$) denotes the real (imaginary) part of $c_{u,j}$. This is an SDP with real variables.

\vspace{10pt}
\noindent\textbf{Examples}

We now show applications of the above SDP method tailored to multipartite systems. As a first example, a three-party system is considered for which Alice possesses a qubit and the other two parties (Bob and Charlie) have no restriction on the dimensionality of the Hilbert spaces. We are able to fully reproduce the bounds obtained in Ref.~\cite{head_leg}. In the next example, we extend Alice's Hilbert space to a qutrit, thereby certifying genuine four-dimensional entanglement. Then we move to a four-party (translationally invariant) Bell scenario, and certify that a Bell value above a certain threshold cannot be obtained with symmetric measurements (that is, when each four parties measure the same observables in the first and second respective settings).

\noindent\emph{I333 inequality} --- We consider the following three-party three-setting permutationally invariant Bell inequality~\cite{head_leg}

\begin{align}
\label{I333}
I_{333}=&\text{sym}\{-P(A_1)-2P(A_3)+P(A_1,B_1)\nonumber\\
&-P(A_1,B_2)+P(A_1,B_3)-2P(A_2,B_2)\nonumber\\
&+2P(A_2,B_3)-2P(A_3,B_3)\}\le 0.
\end{align}
Here we used the short-hand notation $P(A_x,B_y)=p(0,0|x,y)$, $P(A_x)=p(0|x)$, and similarly for the other parties. Notice that the Bell expression above consists of only two-body correlators and single-party marginal terms, which usually provide an advantage in experiments. Such Bell inequalities have been proposed in Ref.~\cite{tura_science} to detect nonlocality in multipartite quantum systems for any number of parties (however, those inequalities involve only two settings per party, hence they can be maximally violated with qubit-systems, unlike the present example).
In eq.~(\ref{I333}), $\text{sym}\{X\}$ means that every term occurring in $X$ should be symmetrized with respect to all possible permutations of the parties, e.g. $\text{sym}\{P(A_1,B_1)\}=P(A_1,B_1)+P(B_1,C_1)+P(A_1,C_1)$.

We next compute upper bounds on the quantum violations assuming different dimensionality of the Hilbert spaces. Lower bound values, on the other hand, are obtained from see-saw iteration in a prior work~\cite{head_leg}. Table~II summarizes the results. Values with an asterisk ($^*$) have been established in the present work. Notation $(D_1D_2D_3)$ refers to the dimensionalities of Alice, Bob and Charlie's Hilbert spaces, respectively. Notice that due to symmetry of the Bell inequality~(\ref{I333}), the same bounds apply to any permutations of $(D_1D_2D_3)$. Establishing upper bound on case $(222)$ with non-degenerate measurements was the most time consuming task, the corresponding SDP problem involved 4894 constraints, and took 3 hours to be solved using Mosek---still the lower bound value has not been saturated. Computing the upper bound for the case $(2\infty\infty)$ required to run the hybrid method (Alice was given level 2 of the qubit hierarchy, whereas Bob and Charlie's system was computed on NPA level 1+AB). In that case, we managed to close the gap between the lower and upper bound values, thereby reproducing the result of Ref.~\cite{head_leg}. We also computed $(3\infty\infty)$ upper bound and recovered the global maximum of 0.1962852 certified by the NPA hierarchy (Alice was given the level 3 of the qutrit hierarchy and took 24 hours for MOSEK to solve the resulting SDP).
Accordingly, any Bell violation of $I_{333}$ bigger than $0.1786897$ cannot be attained with dimensionalities $(2\infty\infty)$ (plus the two other permutations) implying that the underlying three-party state $\rho_{ABC}$ has at least Schmidt number vector $(3,3,3)$ (see e.g. Refs.~\cite{marcus1,marcus2}). Moreover any pure state decomposition of $\rho_{ABC}$ contains at least one state $\sigma_{ABC}=\proj{\psi}$ such that the rank of each single party marginal $\sigma_A$, $\sigma_B$, and $\sigma_C$ is greater than 2. In short, a Bell violation of $I_{333}$ bigger than $0.1786897$ detects in a device-independent way that the three-party state is genuinely three-dimensional entangled.

\begin{table}
\label{G3E}
\begin{center}\begin{tabular}{cccccc}
\hline
   & LB & UB & LB & UB & LB\\
   & (222) & (222) & $(2\infty\infty)$ & $(2\infty\infty)$ & (333)\\
  \hline
  No-deg & $0.0443484$ & $0.0541362^*$ & 0.1783946 & $0.1783946^*$ & 0.1962852\\
  Deg & $0.1783946$ & $0.1783946^*$ & 0.1786897 & 0.1786897 & \\
  \hline
\end{tabular}
\caption{Lower bounds (LB) and upper bounds (UB) on the violation of the
$I_{333}$ inequality in various local dimensions. Values with an asterisk ($^*$) have been established in the present work. The notation $(D_1D_2D_3)$ refers to the dimensionalities of Alice, Bob, and Charlie's Hilbert spaces, respectively. The sign $\infty$ denotes no restriction on dimension of the respective party. Abbrevation Deg/No-deg refers to the situation when Alice has at least one degenerate measurement/all measurements are non-degenerate (i.e., rank-1 projectors). The qutrit value $(333)$ is the overall quantum maximum certified by the NPA hierarchy~\cite{NPAbound}. The upper bound value for $(2\infty\infty)$ in the degenerate case was obtained using the NPA hierarchy as well.}
\end{center}
\end{table}

\noindent\emph{I444 inequality} --- We construct a 3-party Bell inequality which cannot be
violated maximally in state spaces $\C^3\times\C^D\times\C^D$ (and arbitrary permutations thereof)
for any dimension $D$. This extends the previous example to the case when Alice's state space is restricted to a qutrit (instead of a qubit). In particular, the maximal violation is attained in $\C^4\times\C^4\times\C^4$.
Hence, this certifies that the underlying three-party quantum state is genuinely four-dimensional entangled.

Let us consider the following 3-party 4-setting Bell inequality~\cite{persistency}
\begin{equation}
\label{Sineq}
I_{444} = \text{CHSH}_{AB}+\text{CHSH}_{A'C}+\text{CHSH}_{B'C'} \leq 6
\end{equation}
where $A$ and $A'$ denote different sets of measurements for party A, and we use similar notation for parties B and C. In Ref.~\cite{persistency} it has been proved that the maximum quantum violation attainable with biseparable states is $S = 4 + 2\sqrt{2}\approx 6.8284$. Hence, if the above bound is exceeded in a Bell experiment, we can conclude that the state is genuinely tripartite entangled~\cite{horo_rmp,guehne_toth}. The same bound can be derived by using the SDP techniques of Ref.~\cite{BGLP} based on the NPA hierarchy. Below we extend this result to the realm of genuine higher-dimensional entanglement.

To this end, we replace $\text{CHSH}_{B'C'}$ with the Tsirelson bound $2\sqrt 2$~\cite{tsirelson}. This places an upper bound on
$I_{444}$ in~(\ref{Sineq}). Therefore, we are left with optimizing $\text{CHSH}_{AB}+\text{CHSH}_{A'C}+2\sqrt 2$ for $\C^3\times\C^D\times\C^D$ systems, where $D$ denotes arbitrary dimension. To do so, we classify Alice's four observables according to their traces ($\pm 1, \pm 3$) and in each case we can solve the problem with SDP for the hybrid multipartite case. Notice, however, that Bob and Charlie have only two binary-outcome measurements, hence Jordan's Lemma applies and we can assume that Bob and Charlie have traceless qubit observables~\cite{extreme_2222}. Then the problem goes back to upper bounding $\text{CHSH}_{AB}+\text{CHSH}_{A'C}+2\sqrt 2$ in $\C^3\times\C^2\times\C^2$, which can be straightforwardly done using our SDP tools. By running the SDP, the maximum turns out to be $36/7+2\sqrt 2$ up to the numerical precision of the solver Mosek. We also solved the problem assuming that Alice has a qubit yielding the upper bound $2 + 4\sqrt{2}$ up to computer precision. Results are summarized in table~III. All bounds are tight as they are saturated using see-saw search.

\begin{table}
\label{G4E}
\begin{center}\begin{tabular}{cccc}
\hline
  Bisep & $(2\infty\infty)$ & $(3\infty\infty)$ & (444)\\
  \hline
  6.828427 & 7.656854 & 7.971284 & 8.485281\\
  \hline
\end{tabular}
\caption{Maximum quantum bounds on Bell inequality $I_{444}$ in~(\ref{Sineq}) for different local dimensions of Alice. The first column (labeled by Bisep) stands for the case when Alice has a classical system and the other two parties have unrestricted dimensionalities. All bounds are tight as they are matched with lower bounds arising from see-saw iteration.}
\end{center}
\end{table}

Consider now the so-called fully connected Bell state, that is, a 3-party state for which any two parties share a 2-qubit Bell pair,
\begin{equation}
\label{Bell3state}
\ket{\psi_{444}}=\ket{\varphi^+}_{AB}\otimes\ket{\varphi^+}_{A'C}\otimes\ket{\varphi^+}_{B'C'}.
\end{equation}
With this particular $\C^4\times\C^4\times\C^4$ state and measurement settings optimal for CHSH violation, we get the overall quantum maximum of $6\sqrt2\approx 8.485281$ for the Bell inequality~(\ref{Sineq}). By
adding a certain amount of white noise to the state~(\ref{Bell3state}):
\begin{equation}
\rho_\text{noisy}=p\ket{\psi_{444}}\bra{\psi_{444}}+\frac{1-p}{4^3}\id_{4\times 4\times 4},
\label{rhonoisy}
\end{equation}
we get the critical visibility $p_\text{crit} =(36/7 + 2\sqrt 2)/(6\sqrt
2)\approx 0.939425$, above which we can detect the state~(\ref{rhonoisy}) to be genuinely four-dimensional entangled. We believe this threshold is low enough to be interesting from an experimental point of view as well.

\noindent\emph{I2222 with symmetric measurements} --- In Ref.~\cite{tura_jpa}, a search has been conducted for all three- and four-partite binary-outcome Bell inequalities involving two-body correlators that obey translationally symmetry. Any translationally invariant Bell inequality is provably maximally violated by a translationally invariant state when all parties measure the same set of observables (of unlimited dimensionality). Numerical investigations in Ref.~\cite{tura_jpa} suggest that it is not true anymore if we restrict the local Hilbert space dimension of the parties. Let us pick $\#64$ inequality from table~2 in Ref.~\cite{tura_jpa}. Due to the fact that these Bell inequalities involve two dichotomic measurements per site, Jordan's lemma applies and the maximum violation is given by $\beta_Q=6+2\sqrt 2$ in qubit systems. Due to numerics, this value is achieved with different pairs of qubit observables.
Indeed, running our SDP program by building up the bases from random symmetric measurements, we certify that applying the same settings at all sites do not allow us to violate the Bell inequality $\#64$ in table~2 of Ref.~\cite{tura_jpa} (i.e., $\beta_Q^{TI}=\beta_c$ in the notation of the corresponding reference).

%Numerics suggests that in case of inequality $\#64$, qutrit systems with the same settings at all sites also do not suffice to attain $\beta_Q=6+2\sqrt 2$. We pose it as a challenge to certify it using SDP.

\subsection{1-way quantum communication complexity}

Consider the following communication scenario: Alice and Bob are given inputs $x,y$ with probability $p(x,y)$ and have the task to compute the Boolean function $f(x,y)$. To do so, we allow Alice to transmit a $D$-dimensional quantum system to Bob, who, upon receiving it, must make a guess $b$ on $f(x,y)$. We wish to find the strategy which will allow Alice and Bob to maximize the probability that Bob's guess is correct, i.e., $b=f(x,y)$. For example, in a Quantum Random Access Code (QRAC)~\cite{QRAC}, the inputs $\vec{x},y$ can take values in $\{0,1\}^k$ and $\{1,...,k\}$, respectively, and the function to compute is $f(\vec{x},y)=x_y$.

This scenario can be modeled by assuming that Alice prepares a pure quantum state $\rho_x\equiv\proj{\psi_x}\in B(\C^D)$ depending on her input $x$. Bob will conduct a two-outcome projective measurement labeled by $y$, and defined by the projection operators $\{F^y_b:b=0,1\}$, whose outcome will be Bob's guess. In sum, we need to solve the problem

\begin{align}
\max &\sum_{x,y}p(x,y) \tr(\rho_xF^y_{f(x,y)}),\nonumber\\
\mbox{s.t. } &\tr(\rho_x)=1,\rho_x^2=\rho_x, (F^y_b)^2=F^y_b,\nonumber\\
&\rho_x, F^y_b\in B(\C^D).
\label{q_comm}
\end{align}

This problem can be reformulated by assuming that the initial state of Alice's system corresponds to $\rho_{x=0}$ and, for any other input $x$, she sends the state $V_x\rho_0 V_x^\dagger$, where $V_x$ is a unitary operator that can be chosen self-adjoint, i.e., $V_x^2=\id$. The resulting problem belongs to the class~(\ref{npo_dim}), and hence there is a converging SDP hierarchy to attack it. We observed that, in practice, such an SDP hierarchy gave good predictions for $D=2$ at $k=2$. For $D=3$, a third-order relaxation did not suffice to reach the optimal probability of success in $2\to 1$ Quantum Random Access Codes (QRAC)~\cite{QRAC}.

We believe that the main reason for such a slow convergence rate is that the above proposal relies solely on MPIs to enhance dimension constraints. In order to devise a practical SDP hierarchy for problem~(\ref{q_comm}), one needs to find a reformulation of problem~(\ref{q_comm}) where the space $\M^k_{D,\vec{r}}$ is dimension-dependent even for low values of $k$. One such reformulation is immediate: regard $\rho_x$ as rank-1 projectors and assume that the state of the system is the (not normalized) tracial state $\id_D$.

The resulting hierarchy of SDPs should be easy to guess: first, we divide the representations of problem~(\ref{q_comm}) into different classes $r$ depending on the rank of the projectors $\{F^y_0\}$. Second, we generate random states $\rho_x$ and projectors $\{F^y_0\}$ within the class $r$, and, taking the state of the system to be $\id_D$, we use them to build random feasible moment matrices. Those allow us to characterize the space $\M_{D,\vec{r}}^k$. Note that dimension constraints on $\M_{D,\vec{r}}^k$ are present for all $D$ even for $k=1$. E.g.: for any feasible first-order moment matrix $M$, $M_{\id,\id}=D\times M_{\id,\rho_x}$.

The above SDP hierarchy gives good results in practice, but we were not able to prove its convergence. Following Section~\ref{convergence}, it can be shown that, for any class $r$, one can define a representation $\tilde{F}^y_b,\tilde{\rho_x}$ and a tracial state $D\proj{\tilde{\psi}}$ which recover the limiting value of the hierarchy of SDPs. Furthermore, the operator algebra decomposes into a direct integral of representations $z$ with the property that $\mbox{rank}(\tilde{F}^{y}_{b,z})\leq r^y_b$, $\mbox{rank}(\id-\rho_{x,z})\leq D-1$, $\mbox{rank}(\rho_{x,z})\leq 1$. If the dimensionality of $\H_z$ is $D$, that defines a feasible point of problem~(\ref{q_comm}). However, if the dimensionality of $\H_z$ is strictly smaller, we run into trouble: in such representations, $\rho_{x,z}$ can vanish for some values of $x$. Constraints such as $\tr(\rho_z)=1$ can be accounted for by other representations $t$ where $\rho_{x,t}$ is a rank-1 projector, since $D\bra{\tilde{\psi}_t}\rho_{x,t}\ket{\tilde{\psi}_t}=\frac{D}{\text{dim}(\H_t)}>1$.

One possibility to suppress the effects of lower finite-dimensional representations is to add `noncommuting constants', i.e., certain extra operators whose operator relations cannot be realized in dimensions lower than $D$. For instance, in order to guarantee that all representations $Y$ have dimension $D=2$, we could include the Pauli matrices $\sigma_z,\sigma_x$ as operators in the moment matrix $M_k$. With these extra variables, proving convergence can be done by appealing to the convergence of the Lasserre-Parrilo hierarchy~\cite{lasserre,parrilo}. However, we did not find a single situation in our numerical experiments where adding noncommuting constants to fix the dimension was of any advantage.

\vspace{10pt}
\noindent\textbf{Examples}

We explore how the relaxation of the above communication problem performs in practice. To do so, we establish (usually tight) upper bounds in QRAC for various values of $k$ and dimension $D$. We also recompute quantum bounds for the witnesses $I_N$ of Gallego et al.~\cite{gallego}. We further distinguish between real and complex Hilbert spaces and detect general POVM measurements assuming that Alice communicates Bob a quantum system of fixed dimension $D=2$. Note that Ref.~\cite{tavakoli} investigates a generalized QRAC problem where Alice's inputs $\vec{x}$ take values from a string of dits (instead of bit-strings). In that case, our SDP method also showed good performance~\cite{tavakoli}.

\noindent\emph{QRAC} --- We suppose the QRAC has independently and uniformly distributed inputs and Alice is allowed to transmit Bob a $D$-level quantum system. We use the notation of Ref.~\cite{marcin_QRAC} and we denote the average success probability of the optimal \emph{$k\to\log_2(D)$} QRAC by $P_{\max}(k\to \log_2(D))$.

It was known previously from Ref.~\cite{QRAC} that $P_{\max}(2\to 1)=1/2+\sqrt{2}/4$. This is actually the value given by our SDP code at order 2, up to numerical precision. Likewise, when Alice is allowed to transmit a qutrit (case $D=3$), our relaxation based on tracial states at the same order 2 gives $P_{\max}\leq 0.90450850$, which matches with high numerical precision the lower bound value obtained via see-saw technique. Another method from the literature to attack this problem is the Mironowicz-Li-Paw{\l}owski (MLP) SDP hierarchy~\cite{MLP}, whose second-order relaxation gives us the (non-tight) upper bound of 0.9268355.

One can prove that the MLP hierarchy does not converge in general. To do that, first notice that any QRAC can be rewritten as a full-correlation Bell inequality, by defining Alice's observables as $A_x = 2\rho_x - \id$. Then, the only constraint that the MLP hierarchy adds to the NPA hierarchy is that $\ve{A_x} = 2-D$. Taking $D=2$, we see that the problem of calculating $P_{\max}(k\to 1)$ reduces to maximizing the violation of a full-correlation Bell inequality constraining (some of) its marginals to be uniform. By Tsirelson's theorem, the maximum is anyway attained when the marginals are uniform, so the constraint is automatically satisfied~\cite{tsirelson,avis09}. This means we can simply solve Tsirelson's SDP to find out this maximum and the minimal dimension necessary to attain it~\cite{avis09}. For $k=4$, we see that $D=4$ is necessary to reach the maximum, so the hierarchy did not converge to the maximum for $D=2$.

By increasing the dimension $D$ and the parameter $k$, the 2nd-order relaxation of the hierarchy based on the tracial states also performs well. The entries in the first three rows of table~IV are from Ref.~\cite{finite_dim_short}, which shows lower and upper bounds on the average success probability for QRAC $k\to \log_2(D)$ for $k=3$ and for different values of $D$. The upper bounds (UB) are computed via our SDP in a normal desktop, and took less than 1 hour for any of the $D$ values (assuming a given rank-combination of measurements) using the solver SeDuMi~\cite{sedumi}. The upper bounds (UB') are resulting from the second-order relaxation of the MLP method~\cite{MLP}. We also show upper bounds (UB'') derived from the Moroder et al.~\cite{moroder} hierarchy by fixing negativity $(D-1)/2$ and adding the constraint $P(a|x)=1/D$ on Alice's marginal distributions. As the table~\ref{tab:QRAC3} shows, except for $D=2,4$, where the outputs of all methods coincide, the new tool gives predictions $\sim 10^{-2}$ more accurate than the MLP method and the method based on Moroder et al. hierarchy.

Let us pick $P_{\max}(3\to 1)=0.788675$ from table~\ref{tab:QRAC3}. This is precisely the value given by the construction of Ike Chuang~\cite{QRAC} proving optimality of the \emph{complex} qubit value. However, we can apply in this case the same ideas to characterize the properties of \emph{real} qubit systems as well. By generating the basis from randomly chosen real-valued qubit states $\rho_x$ and projectors $\{F^y_0\}$, we get the (tight) upper bound 0.7696723. Hence, this simple example allows us to distinguish between real and complex two-level systems.

\begin{table}
\begin{center}\begin{tabular}{ccccccc}
\hline
D& 2& 3& 4& 5& 6& 7\\
\hline
LB& 0.788675& 0.832273& 0.908248& 0.924431& 0.951184& 0.969841\\
UB& 0.788675& 0.832273& 0.908248& 0.924445& 0.954123& 0.969841\\
UB'& 0.788675& 0.853553& 0.908248& 0.934264& 0.957785& 0.979567\\
UB''&0.788675& 0.852156& 0.908248& 0.931201& 0.954140& 0.977072\\
  \hline
\end{tabular}
\caption{Lower (LB) and various upper bounds (UB, UB', UB'') on $P_{\max}(3\to\log_2(D))$ detailed in the text.}
\label{tab:QRAC3}
\end{center}
\end{table}

\noindent\emph{$I_N$ family} --- As another example, we used the SDP program based on tracial states to re-compute the maximal quantum value of the prepare-and-measure dimension witnesses $I_N$ defined in Ref.~\cite{gallego}, table I. The second relaxation of the SDP hierarchy turns out to produce upper bounds for cases $N=3,4$ and $D=2,3$ which match the lower bounds obtained with see-saw method. Let us note that the conclusions of the experimental paper~\cite{gallego_exp} relied on the conjecture that the inequality $I_4\leq 7.9689$ cannot be violated by quantum systems of dimension $D=3$.

In a recent experimental paper~\cite{Ambrosio}, the $I_N$ dimension witness has been investigated for $N=7$. Using a heuristic search, lower bounds are provided for dimensions $D=2,\ldots,6$, which were conjectured to be optimal. First row (LB) in our table~\ref{tab:I7} shows lower bound results of Ref.~\cite{Ambrosio}, which match our lower bounds (LB' in second row) except the case $D=3$, where we got a slightly higher lower bound value. This also justifies the need to relaxation methods providing certified upper bound values. Remarkably, our SDP upper bound values (UB in third row) are close to the LB' values differing in the worst case $D=3$ in the 2nd digit. For instance, Ref.~\cite{Ambrosio} provides the experimental value of $I_7=25.44\pm 0.02$ for the preparation of 6 dimensional states. Therefore our UB value of 24.8991 for $D=5$ certifies the generation of at least 6 dimensional quantum states.

\begin{table}
\begin{center}\begin{tabular}{cccccc}
\hline
D& 2& 3& 4& 5& 6\\
\hline
LB& 17.3976& 20.7143& 23.2167& 24.8978& 26.1017\\
LB'& 17.3976& 20.7085& 23.2167& 24.8987& 26.1017\\
UB& 17.3976& 20.7718& 23.2180& 24.8991& 26.1019\\
\hline
\end{tabular}
\caption{Lower and upper bounds on $I_7$ witness for dimensions $D=2,\ldots,6$. LB stands for the data in Ref.~\cite{Ambrosio}, table I, whereas LB' is due to our see-saw technique and UB is resulting from our SDP method.}
\label{tab:I7}
\end{center}
\end{table}

\noindent\emph{POVM witnesses} --- N. Gisin~\cite{open} asked if there exists a Bell inequality which requires POVMs for optimal violation on some quantum state. This question has been answered affirmatively in case of two-qubit states (see e.g., Refs.~\cite{VB10,barra,acin2015optimal,kleinmann2015correlations}). However, the question is still open for Bell inequalities defining facet of the local polytope (though, numerical study suggests the existence of such cases for three parties and high dimensional states~\cite{grandjean}). We pose a similar question in the prepare-and-measure communication scenario: By fixing dimension (say, Alice is allowed to send Bob a two-level system), are there witnesses which allow higher violations when Bob performs general POVM measurements instead of standard projective measurements? Our SDP tools allow one to certify the existence of such POVM witnesses.

To this end, we pick the $V_4$ witness of Ref.~\cite{ours} and consider dimension $D=2$. This witness consists of four preparations ($x=1,\ldots,4$) on Alice's side and six measurements ($y=1,\ldots,6$) on Bob's side. For $D=2$, the maximum value of $V_4$ equals $2\sqrt 6$, which can be attained if Alice prepares four states pointing toward the vertices of the regular tetrahedron (corresponding to SIC POVM elements). We add a four-outcome measurement to Bob ($y=7$) with four outcomes $b=1,2,3,4$ to the original $V_4$ witness and define the modified $V_4$ witness as follows
\begin{equation}
\label{V4prime}
V_4'=V_4-\sum_{i=1}^4P(b=i|x=i,y=7)\le 2\sqrt 6\simeq 4.8990.
\end{equation}
We remark that a similar modification was used in the context of Bell nonlocality in Ref.~\cite{acin2015optimal}.
As the last term in the inequality cannot be positive, the qubit bound $2\sqrt 6\simeq 4.8990$ using POVMs follows from the bound on $V_4$. Indeed, by using the known optimal qubit settings for $V_4$, the only way of getting the maximal violation of $2\sqrt 6$ is when Bob's POVM elements in setting $y=7$ are anti-aligned with the four tetrahedron states prepared by Alice, so that all probabilities $P(b=i|x=i,y=7)$ becomes zero. By assuming projective qubit measurements for Bob and running SDP in case of $D=2$, we obtain $2(\sqrt 2 +1)\simeq 4.8284$ up to numerical precision on the witness $V_4'$ in Eq.~(\ref{V4prime}). Hence, any value bigger than $4.8284$ for $V_4'$ certifies in a semi-device-independent way that Bob's measurement $y=7$ was in fact a general POVM measurement.

\section{Some tips on implementation}
\label{tips}
In this section, we offer some tips to implement the programs defined above. As we will see, those tricks will make the hierarchy for NPO under dimension constraints much easier to code and modify than its dimension-free counterpart~\cite{siam}. For simplicity, we will assume that our program does not involve localizing matrices, i.e., that all polynomial restrictions appear as identities. We will also presume that operator representations can be divided into classes $\vec{r}$, and that generating a random instance of each class can be done efficiently.

Firstly, we need a subroutine {\tt [X,rho]=genSamp(D,n,r)} that generates a cell-array $X$ of random operators $X_0,X_1,...,X_{n}$, with $X_0=\id_D$ and $X_1,...,X_{n}\in B(\C^D)$ satisfying the appropriate class constraints, determined by the vector $r$. {\tt genSamp} must also return a quantum state $\rho\in B(\C^D)$ (in our examples, either a pure random state $\proj{\psi}$, the maximally entangled state or the unnormalized maximally mixed state).

Secondly, we need a subroutine {\tt G = buildG(X,rho,k)} that, given the cell-array $X$ and an index $k$, generates a $k^{th}$-order moment matrix of the form:

\be
G_{\vec{i},\vec{j}}=\tr(\rho X_{i_k}^\dagger...X_{i_1}^\dagger X_{j_1}...X_{j_k}),
\ee

\noindent where $\vec{i},\vec{j}\in\{0,...,n\}^k$, if the variables $X_1,...,X_n$ are Hermitian, or $\vec{i},\vec{j}\in\{0,...,2n\}^k$ if they are not. In either case, monomials of $G$ can be accessed via $\tr(G\ket{\vec{i}}\bra{\vec{j}})$, choosing $\vec{i},\vec{j}$ appropriately. Note that, in this representation, different columns of $G$ correspond to the same operator, e.g.: $01$ and $10$. That does not matter, because, at the end of the day, redundant columns will be suppressed by the matrix $V$ mapping the space where $G$ is defined to the support of $\mbox{span}\{G\}$, see Remark~\ref{support}.

If the reader is an Octave or MATLAB user, the following considerations will lead to very fast code.

Any $D\times D$ matrix $A=\sum_{i,j=1}^D A_{i,j}\ket{i}\bra{j}$ can be represented in vector form as $\ket{A}=\sum_{i,j} A_{i,j}\ket{i}\bra{j}$: to go from $A$ to $\ket{A}$, one can invoke the in-built function {\tt reshape}. It can be verified that $(\id_D\otimes\bra{\psi^+}\otimes \id_D)\ket{A}\ket{B}=\ket{AB}$, where $\ket{\psi^+}$ is the non-normalized maximally entangled state $\ket{\psi^+}=\sum_{i=1}^D\ket{i,i}$. Similarly, $(\id_D\otimes \bra{\phi})\ket{A}=A\ket{\phi^*}$, and, hence, $\bra{A}(\id_D\otimes\rho^*)\ket{B}=\tr(A^\dagger B\rho)$, for any Hermitian matrix $\rho\in B(C^D)$.

Now, for random operators $X_1,...,X_n$ in the class $r$, define $\Lambda=\sum_{i=0}^n\ket{X_i}\bra{i}$. From the above, it follows that

\begin{multline}
\left(\id_D\otimes\bra{\psi^+}\otimes\id_D\right)\left(\Lambda\otimes \sum_{j_2,...,j_l}\ket{X_{j_2}...X_{j_l}}\bra{j_2,...j_l}\right)\\
= \sum_{j_1,...,j_l}\ket{X_{j_1}...X_{j_l}}\bra{j_1,...i_l}.
\end{multline}

\noindent By tensoring $\Lambda$ sequentially and projecting on the maximally entangled state, we thus (quickly) obtain the operator

\be
C\equiv\sum_{j_1,...,j_k}\ket{X_{j_1}...X_{j_k}}\bra{j_1,...i_k}.
\ee

\noindent Then it can be verified that

\be
G=C^\dagger (\id_D\otimes \rho^*) C.
\ee

Finally, we must code a third subroutine {\tt [basisMat,V]=buildBasis(D,n,r,k)} that, by repeatedly calling {\tt genSamp} and {\tt buildG}, derives an orthonormal matrix basis for $\mbox{span} \{G\}$ and the isometry $V$ described in Remark~\ref{support}. From there, it is straightforward to implement program~(\ref{npo_rel2}).

Notice that all operator relations are determined by {\tt genSamp} only. This means, for example, that if we wish to optimize over projection (unitary operators) all we need to do is program {\tt genSamp} to generate random tuples of projection (unitary) operators $P_1,...,P_n$ ($U_1,...,U_n,U^\dagger_1,...,U^\dagger_n$). Switching from one type of polynomial constraints to another can thus be done straightforwardly with the above implementation.

\section{Conclusion}

In this paper we have extended the notion of noncommutative polynomial optimization (NPO) to scenarios where the dimensionality of the spaces where the noncommuting variables act is bounded from above. We have presented a complete hierarchy of SDP relaxations to solve such problems and we have explored its performance by applying it to solve a number of open problems in quantum information theory.

Our research raises several questions which deserve further study. The first one is whether the hierarchy of relaxations proposed to study 1-way quantum complexity is complete or not. As we showed, the main obstacle to prove convergence lies in interference from low-dimensional operator representations with `dark states'. For $D=2$, the extreme points of such additional one-dimensional representations are finite, and thus the problem of determining whether the hierarchy converges in a given prepare-and-measure scenario amounts to proving that all such points can be reproduced by two-level quantum systems. Similarly, the convergence of the sequence of relaxations to the maximum value of a specific functional for arbitrary $D$ can be verified by establishing an upper bound for the dark-state value smaller than or equal to the value of a concrete $D$-dimensional realization. It would be more satisfactory, though, to have a general convergence result.

Another open problem is whether the relaxations proposed for optimizations over finite-dimensional \emph{real} operator algebras actually converge. Here we are faced with the problem that certain operator representations, irreducible in the real space, can be expressed as a direct sum of non-trivial representations if we allow complex unitary transformations. Hence, if we wished to follow the proof for complex algebras, we would encounter problems at the step of applying von Neumann's direct-integral theorem~\cite{von_neumann}.

Finally, note that in this work we have not exploited the symmetry of the functionals to optimize. If our aim is to bound quantum nonlocality under dimension constraints, that leaves us with a method that, in the best scenario, would allow us to conduct optimizations for $D=2,3$ with a normal computer. Which dimensionalities could become accessible if we chose to play with symmetries is an intriguing question.

\section*{Acknowledgements}
M.N. acknowledges support by Spanish MINECO Project No. FIS2013-40627-P.
T.V. acknowledges support from the Hungarian National Research Fund OTKA (K111734) and a J\'anos Bolyai Grant of
the Hungarian Academy of Sciences. M.A. and A.F. acknowledge support from the European Commission project RAQUEL (No.~323970) and the Austrian Science Fund (FWF) through the Special Research Programme FoQuS, Individual Project  (No.~2462) and the Doctoral Programme CoQuS.

\begin{appendix}
\section{Type II and III factors violate the standard MPI}

The purpose of this Appendix is to prove its title. To do so, we will require the following lemma:

\begin{lemma}
Let $\A$ be a type II or III factor. Then, for all $n$, there exists non-trivial projectors $\{P_i\}_{i=1}^n\subset \A$ such that
\begin{enumerate}
\item
$\sum_{i=1}^n P_i=\id$.
\item
$P_1\A P_2\A...P_n\A\not= 0$.

\end{enumerate}
\end{lemma}

\begin{proof}
Suppose that the statement is not true for general $n$ and let $N$ be the greatest number such that it holds. Note that $N>1$. Indeed, if $P_1\A P_2=0$, for non-trivial $P_1,P_2$ then, for all $x\in\A$,

\begin{align}
&[x,P_1]=(P_1+P_2)[x,P_1](P_1+P_2)=\nonumber\\
&=P_1xP_1+P_2xP_1-P_1xP_1-P_1xP_2=0.
\end{align}

\noindent That is impossible, because factors, by definition are central, i.e., the only elements of $\A$ commuting will $\A$ are multiples of the identity.

Now, suppose that $\{P_i\}_{i=1}^N$ satisfy the conditions of the lemma. Then we can always write $P_N=P_N'+P_{N+1}'$, where $P_N',P_{N+1}'$ are non-zero projectors such that $P_1'\A P_2'\A...P_N'\A\not=0$. Call $\B$ the algebra generated by $\A P_1'\A P_2'\A...P_N'\A$, and let $\Pi$ denote its identity, i.e., a projector $\Pi\in \B$ such that $\Pi y=y$ for all $y\in \B$. Note that, due to the definition of $\B$, $\A\B=\B\A=\B$, and hence $x\Pi,\Pi x\in\B$ for all $x\in \A$. It follows that $\Pi x=\Pi x\Pi=x\Pi$, that is $[x,\Pi]=0$.

On the other hand, $N$ is the greatest number such that the conditions of the lemma hold, and so $\A P_1'...\A P_N'\A P_{N+1}'=P_{N+1}'\A P_1'...\A P_N'\A=0$. In other words: $\B P_{N+1}'=0$, and, consequently, $\Pi P_{N+1}'=0$. We conclude that $\Pi$, which is neither $0$ nor the identity, must commute with all $x\in\A$, contradicting the centrality of $\A$.

\end{proof}

Now, given $n$ and a type II or III factor $\A$, choose orthogonal projectors $\{P_i\}_{i=1}^n\subset \A$ and operators $x_1,...,x_n\in\A$ such that $P_1x_1P_2x_2...P_n\not=0$. Then define the operators $E_i\equiv P_ix_iP_{i+1}$, for $i=1,...,n-1$ and compute the fundamental polynomial $I_{n-1}(E)$. Due to the orthogonality of $\{P_i\}$, the only non-vanishing product is $E_1E_2...E_{n-1}=P_1x_1P_2x_2...P_n\not=0$. $\A$ hence violates $I_{n-1}=0$.

\end{appendix}

\bibliographystyle{unsrt}
\bibliography{fin_dim2}

\begin{thebibliography}{10}

\bibitem{review_comm}
A.~Yao.
\newblock Quantum circuit complexity.
\newblock {\em Proceedings of the 34th IEEE FOCS, IEEE, Los Alamitos, CA}, page
  352–360, 1993.

\bibitem{review_comm2}
H.~Buhrman, R.~Cleve, S.~Massar, and R.~de~Wolf.
\newblock Nonlocality and communication complexity.
\newblock {\em Rev. Mod. Phys.}, 82:665, 2010.

\bibitem{bell}
John~S. Bell.
\newblock On the einstein podolsky rosen paradox.
\newblock {\em Physics}, 1(3):195--200, 1964.

\bibitem{d_wit1}
N.~Brunner, S.~Pironio, A.~Ac\'in, N.~Gisin, A.~A. M\'ethot, and V.~Scarani.
\newblock Testing the dimension of hilbert spaces.
\newblock {\em Phys. Rev. Lett.}, 100:210503, 2008.

\bibitem{d_wit2}
K.F. P\'al and T.~V\'ertesi.
\newblock Efficiency of higher dimensional hilbert spaces for the violation of
  bell inequalities.
\newblock {\em Phys. Rev. A}, 77:042105, 2008.

\bibitem{d_wit3}
S.~Wehner, M.~Christandl, and A.~C. Doherty.
\newblock Lower bound on the dimension of a quantum system given measured data.
\newblock {\em Phys. Rev. A}, 78:062112, 2008.

\bibitem{d_wit4}
J.~Bri\"{e}t, H.~Buhrman, and B.~Toner.
\newblock A generalized grothendieck inequality and nonlocal correlations that
  require high entanglement.
\newblock {\em Comm. Math. Phys.}, 305:827, 2011.

\bibitem{distill}
C.H. Bennett, G.~Brassard, S.~Popescu, B.~Schumacher, J.~A. Smolin, and W.~K.
  Wooters.
\newblock Purification of noisy entanglement and faithful teleportation via
  noisy channels.
\newblock {\em Phys. Rev. Lett.}, 76:722--725, 1996.

\bibitem{chsh}
John~F. Clauser, Michael~A. Horne, Abner Shimony, and Richard~A. Holt.
\newblock Proposed experiment to test local hidden-variable theories.
\newblock {\em Phys. Rev. Lett.}, 23:880--884, Oct 1969.

\bibitem{tsirel_pr}
L.~A. Khalfin and B.~S. Tsirelson.
\newblock Quantum and quasi-classical analogues of bell inequalities.
\newblock {\em Symposium on the Foundations of Modern Physics}, pages 441--460,
  1985.

\bibitem{pr}
Sandu Popescu and Daniel Rohrlich.
\newblock Quantum nonlocality as an axiom.
\newblock {\em Foundations of Physics}, 24:379--385, 1994.

\bibitem{NPAbound}
Miguel Navascu\'es, Stefano Pironio, and Antonio Ac\'{i}n.
\newblock Bounding the set of quantum correlations.
\newblock {\em Phys. Rev. Lett.}, 98:010401, 2007.

\bibitem{npa}
M.~Navascu\'es, S.~Pironio, and A.~Ac\'{i}n.
\newblock A convergent hierarchy of semidefinite programs characterizing the
  set of quantum correlations.
\newblock {\em New J. Phys.}, 10:073013, 2008.

\bibitem{siam}
M.~Navascu\'{e}s S.~Pironio and A.~Ac\'{i}n.
\newblock Convergent relaxations of polynomial optimization problems with
  noncommuting variables.
\newblock {\em SIAM J. Optim.}, 20(5):2157--2180, 2010.

\bibitem{finite_dim_short}
Miguel Navascu\'es and Tam\'as V\'ertesi.
\newblock Bounding the set of finite dimensional quantum correlations.
\newblock {\em Phys. Rev. Lett.}, 115:020501, Jul 2015.

\bibitem{VB10}
T.~V\'ertesi and E.~Bene.
\newblock Two-qubit bell inequality for which positive operator-valued
  measurements are relevant.
\newblock {\em Phys. Rev. A}, 82:062115, Dec 2010.

\bibitem{barra}
J.~F. Barra, E.~S. G\'omez, G.~Ca\~nas, W.~A.~T. Nogueira, L.~Neves, and
  G.~Lima.
\newblock Higher quantum bound for the v\'ertesi-bene-bell inequality and the
  role of positive operator-valued measures regarding its threshold detection
  efficiency.
\newblock {\em Phys. Rev. A}, 86:042114, Oct 2012.

\bibitem{temp_corr}
C.~Budroni, T.~Moroder, M.~Kleinmann, and O.~G\"{u}hne.
\newblock Bounding temporal correlations.
\newblock {\em Phys. Rev. Lett.}, 111:020403, 2013.

\bibitem{MPI}
C.~Procesi.
\newblock {\em Rings with polynomial identities}.
\newblock Marcel Dekker, New York, 1973.

\bibitem{sdp}
L.~Vandenberghe and S.~Boyd.
\newblock Semidefinite programming.
\newblock {\em SIAM Review}, 38:49, 1996.

\bibitem{mps_amazing}
M.~Navascu\'es and T.~V\'ertesi.
\newblock In preparation.

\bibitem{reed_simon}
M.~Reed and B.~Simon.
\newblock {\em Functional Analysis}.
\newblock New York: Academic, 1980.

\bibitem{von_neumann}
J.~von Neumann.
\newblock On rings of operators. reduction theory.
\newblock {\em The Annals of Mathematics 2nd Ser.}, 50 (2):401–485, 1949.

\bibitem{takesaki}
M.~Takesaki.
\newblock {\em Theory of Operator Algebras I}.
\newblock Springer, 1979.

\bibitem{lasserre}
J.~B. Lasserre.
\newblock Global optimization with polynomials and the problem of moments.
\newblock {\em SIAM J. Optim.}, 11:796–817, 2001.

\bibitem{parrilo}
P.~Parrilo.
\newblock {\em Structured Semidefinite Programs and Semialgebraic Geometry
  Methods in Robstness and Optimization, Ph.D. thesis}.
\newblock California Institute of Technology, Pasadena, CA, 2000.

\bibitem{mosek}
L.~Vandenberghe and S.~Boyd.
\newblock {\em The MOSEK optimization toolbox for MATLAB manual. Version 7.0
  (Revision 140).}
\newblock MOSEK ApS, Denmark.

\bibitem{sedumi}
Jos~F Sturm.
\newblock Using sedumi 1.02, a matlab toolbox for optimization over symmetric
  cones.
\newblock {\em Optimization methods and software}, 11(1-4):625--653, 1999.

\bibitem{yalmip}
J.~L\"{o}fberg.
\newblock Yalmip : A toolbox for modeling and optimization in matlab.
\newblock In {\em Proceedings of the CACSD Conference}, Taipei, Taiwan, 2004.

\bibitem{i3322}
Daniel Collins and Nicolas Gisin.
\newblock A relevant two qubit bell inequality inequivalent to the chsh
  inequality.
\newblock {\em Journal of Physics A: Mathematical and General}, 37(5):1775,
  2004.

\bibitem{head_leg}
Miguel Navascu\'es, Gonzalo de~la Torre, and Tam\'as V\'ertesi.
\newblock Characterization of quantum correlations with local dimension
  constraints and its device-independent applications.
\newblock {\em Phys. Rev. X}, 4:011011, Jan 2014.

\bibitem{moroder}
Tobias Moroder, Jean-Daniel Bancal, Yeong-Cherng Liang, Martin Hofmann, and
  Otfried G{\"u}hne.
\newblock Device-independent entanglement quantification and related
  applications.
\newblock {\em Phys. Rev. Lett.}, 111(3):030501, 2013.

\bibitem{negativity}
G.~Vidal and R.~F. Werner.
\newblock Computable measure of entanglement.
\newblock {\em Phys. Rev. A}, 65:032314, Feb 2002.

\bibitem{vidick}
Thomas Vidick and Stephanie Wehner.
\newblock More nonlocality with less entanglement.
\newblock {\em Phys. Rev. A}, 83:052310, May 2011.

\bibitem{zoo}
Ben Lang, Tam{\'a}s V{\'e}rtesi, and Miguel Navascu{\'e}s.
\newblock Closed sets of correlations: answers from the zoo.
\newblock {\em Journal of Physics A: Mathematical and Theoretical},
  47(42):424029, 2014.

\bibitem{palvert}
K\'aroly~F. P\'al and Tam\'as V\'ertesi.
\newblock Quantum bounds on bell inequalities.
\newblock {\em Phys. Rev. A}, 79:022120, 2009.

\bibitem{seesaw1}
R.~F. Werner and M.~M. Wolf.
\newblock All-multipartite bell-correlation inequalities for two dichotomic
  observables per site.
\newblock {\em Phys. Rev. A}, 64:032112, 2001.

\bibitem{seesaw2}
K.F. P\'al and T.~V\'ertesi.
\newblock Maximal violation of the i3322 inequality using infinite dimensional
  quantum systems.
\newblock {\em Phys. Rev. A}, 82:022116, 2010.

\bibitem{faacets}
Denis Rosset, Jean-Daniel Bancal, and Nicolas Gisin.
\newblock Classifying 50 years of bell inequalities.
\newblock {\em Journal of Physics A: Mathematical and Theoretical},
  47(42):424022, 2014.

\bibitem{VP09}
Tam\'as V\'ertesi and K\'aroly~F. P\'al.
\newblock Bounding the dimension of bipartite quantum systems.
\newblock {\em Phys. Rev. A}, 79:042106, Apr 2009.

\bibitem{nielsen_chuang}
M.~A. Nielsen and I.~L. Chuang.
\newblock {\em Quantum computation and quantum information}.
\newblock Cambridge University Press, UK, 2000.

\bibitem{pironio}
Stefano Pironio.
\newblock All clauser--home--shimony--holt polytopes.
\newblock {\em Journal of Physics A: Mathematical and Theoretical},
  47(42):424020, 2014.

\bibitem{latorre}
A.~Ac\'{i}n, T.~Durt, N.~Gisin, and J.~I. Latorre.
\newblock Quantum nonlocality in two three-level systems.
\newblock {\em Phys. Rev. A}, 65:052325, May 2002.

\bibitem{tura_science}
Jordi Tura, R~Augusiak, AB~Sainz, T~V{\'e}rtesi, M~Lewenstein, and A~Ac{\'\i}n.
\newblock Detecting nonlocality in many-body quantum states.
\newblock {\em Science}, 344(6189):1256--1258, 2014.

\bibitem{marcus1}
Marcus Huber, Florian Mintert, Andreas Gabriel, and Beatrix~C. Hiesmayr.
\newblock Detection of high-dimensional genuine multipartite entanglement of
  mixed states.
\newblock {\em Phys. Rev. Lett.}, 104:210501, May 2010.

\bibitem{marcus2}
C.~{Spengler}, M.~{Huber}, A.~{Gabriel}, and B.~C. {Hiesmayr}.
\newblock {Examining the dimensionality of genuine multipartite entanglement}.
\newblock June 2011.

\bibitem{persistency}
Nicolas Brunner and Tam\'as V\'ertesi.
\newblock Persistency of entanglement and nonlocality in multipartite quantum
  systems.
\newblock {\em Phys. Rev. A}, 86:042113, Oct 2012.

\bibitem{horo_rmp}
Ryszard Horodecki, Pawe\l{} Horodecki, Micha\l{} Horodecki, and Karol
  Horodecki.
\newblock Quantum entanglement.
\newblock {\em Rev. Mod. Phys.}, 81:865--942, Jun 2009.

\bibitem{guehne_toth}
Otfried G{\"u}hne and G{\'e}za T{\'o}th.
\newblock Entanglement detection.
\newblock {\em Physics Reports}, 474(1):1--75, 2009.

\bibitem{BGLP}
Jean-Daniel Bancal, Nicolas Gisin, Yeong-Cherng Liang, and Stefano Pironio.
\newblock Device-independent witnesses of genuine multipartite entanglement.
\newblock {\em Phys. Rev. Lett.}, 106:250404, Jun 2011.

\bibitem{tsirelson}
B.S. Cirel'son.
\newblock Quantum generalizations of bell's inequality.
\newblock {\em Letters in Mathematical Physics}, 4(2):93--100, 1980.

\bibitem{extreme_2222}
Ll. Masanes.
\newblock Extremal quantum correlations for n parties with two dichotomic
  observables per site.
\newblock 2005.

\bibitem{tura_jpa}
J~Tura, AB~Sainz, T~V{\'e}rtesi, A~Ac{\'\i}n, M~Lewenstein, and R~Augusiak.
\newblock Translationally invariant multipartite bell inequalities involving
  only two-body correlators.
\newblock {\em Journal of Physics A: Mathematical and Theoretical},
  47(42):424024, 2014.

\bibitem{QRAC}
Andris Ambainis, Ashwin Nayak, Amnon Ta-Shma, and Umesh Vazirani.
\newblock Dense quantum coding and quantum finite automata.
\newblock {\em Journal of the ACM (JACM)}, 49(4):496--511, 2002.

\bibitem{gallego}
Rodrigo Gallego, Nicolas Brunner, Christopher Hadley, and Antonio Ac{\'\i}n.
\newblock Device-independent tests of classical and quantum dimensions.
\newblock {\em Phys. Rev. Lett.}, 105(23):230501, 2010.

\bibitem{tavakoli}
Armin Tavakoli, Alley Hameedi, Breno Marques, and Mohamed Bourennane.
\newblock Quantum random access codes using single $d$-level systems.
\newblock {\em Phys. Rev. Lett.}, 114:170502, Apr 2015.

\bibitem{marcin_QRAC}
Hong-Wei Li, Marcin Paw{\l}owski, Zhen-Qiang Yin, Guang-Can Guo, and Zheng-Fu
  Han.
\newblock Semi-device-independent randomness certification using
  n\ensuremath{\rightarrow}1 quantum random access codes.
\newblock {\em Phys. Rev. A}, 85(5):052308, 2012.

\bibitem{MLP}
Piotr Mironowicz, Hong-Wei Li, and Marcin Paw\l{}owski.
\newblock Properties of dimension witnesses and their semidefinite programming
  relaxations.
\newblock {\em Phys. Rev. A}, 90:022322, Aug 2014.

\bibitem{avis09}
D.~{Avis}, S.~{Moriyama}, and M.~{Owari}.
\newblock {From Bell Inequalities to Tsirelson's Theorem}.
\newblock {\em IEICE Trans. on Fundamentals}, 92:1254--1267, 2009.

\bibitem{gallego_exp}
Martin Hendrych, Rodrigo Gallego, Michal Mi{\v{c}}uda, Nicolas Brunner, Antonio
  Ac{\'\i}n, and Juan~P Torres.
\newblock Experimental estimation of the dimension of classical and quantum
  systems.
\newblock {\em Nature Physics}, 8(8):588--591, 2012.

\bibitem{Ambrosio}
Vincenzo D'Ambrosio, Fabrizio Bisesto, Fabio Sciarrino, Johanna~F. Barra,
  Gustavo Lima, and Ad\'an Cabello.
\newblock Device-independent certification of high-dimensional quantum systems.
\newblock {\em Phys. Rev. Lett.}, 112:140503, Apr 2014.

\bibitem{open}
Nicolas Gisin.
\newblock Bell inequalities: many questions, a few answers.
\newblock In {\em Quantum Reality, Relativistic Causality, and Closing the
  Epistemic Circle}, pages 125--138. Springer, 2009.

\bibitem{acin2015optimal}
Antonio Ac{\'\i}n, Stefano Pironio, Tam{\'a}s V{\'e}rtesi, and Peter Wittek.
\newblock Optimal randomness certification from one entangled bit.
\newblock 2015.

\bibitem{kleinmann2015correlations}
Matthias Kleinmann and Adan Cabello.
\newblock Correlations in nature are not merely quantum dichotomic.
\newblock 2015.

\bibitem{grandjean}
Basile Grandjean, Yeong-Cherng Liang, Jean-Daniel Bancal, Nicolas Brunner, and
  Nicolas Gisin.
\newblock Bell inequalities for three systems and arbitrarily many measurement
  outcomes.
\newblock {\em Phys. Rev. A}, 85:052113, May 2012.

\bibitem{ours}
Nicolas Brunner, Miguel Navascu\'es, and Tam\'as V\'ertesi.
\newblock Dimension witnesses and quantum state discrimination.
\newblock {\em Phys. Rev. Lett.}, 110:150501, Apr 2013.

\bibitem{mod_GS}
\r{A}. Bj\"{o}rck and C.~C. Paige.
\newblock Loss and recapture of orthogonality in the modified gram–schmidt
  algorithm.
\newblock {\em SIAM. J. Matrix Anal. \& Appl.}, 13(1):176–190, 1992.

\end{thebibliography}

\end{document}